\documentclass[11pt]{article}

\usepackage{arxiv}

\usepackage[utf8]{inputenc}
\usepackage[T1]{fontenc}
\usepackage{hyperref}
\usepackage{url}
\usepackage{booktabs}
\usepackage{amsfonts}
\usepackage{nicefrac}
\usepackage{microtype}
\usepackage{graphicx}
\usepackage{natbib}
\usepackage{doi}
\usepackage{float,setspace}
\usepackage{tikz}
\usetikzlibrary{
  arrows.meta,
  positioning,
  decorations.pathreplacing
}

\usepackage{amssymb,amsthm,amsmath}
\usepackage{xcolor,paralist}
\usepackage{enumitem}

\usepackage[capitalise,noabbrev]{cleveref}

\newtheorem{theorem}{Theorem}[section]
\newtheorem{lemma}[theorem]{Lemma}

\newtheorem{definition}[theorem]{Definition}

\newtheorem{remark}[theorem]{Remark}
\newtheorem{corollary}[theorem]{Corollary}

\newcommand{\Zq}{\mathbb{Z}_q}
\newcommand{\F}{\mathcal{F}}
\newcommand{\E}{\mathcal{E}}
\newcommand{\C}{\mathcal{C}}
\newcommand{\R}{\mathcal{R}}
\newcommand{\A}{\mathcal{A}}

\newcommand{\Sc}{\mathcal{S}}

\newcommand{\Hone}{\ensuremath{\mathsf{H}_1}}
\newcommand{\Htwo}{\ensuremath{\mathsf{H}_2}}
\newcommand{\Com}{\ensuremath{\mathsf{Com}}}
\newcommand{\NIZK}{\ensuremath{\mathsf{NIZK}}}
\newcommand{\EvalCom}{\ensuremath{\mathsf{EvalCom}}}
\newcommand{\Sim}{\ensuremath{\mathsf{Sim}}}
\newcommand{\Ext}{\ensuremath{\mathsf{Ext}}}
\newcommand{\negl}{\ensuremath{\mathsf{negl}}}
\newcommand{\poly}{\ensuremath{\mathsf{poly}}}
\newcommand{\Prove}{\ensuremath{\mathsf{Prove}}}
\newcommand{\Verify}{\ensuremath{\mathsf{Verify}}}
\newcommand{\Setup}{\ensuremath{\mathsf{Setup}}}

\newcommand{\cert}{\ensuremath{\mathsf{cert}}}
\newcommand{\scheme}{\ensuremath{\mathsf{PIVOT}}}
\newcommand{\Real}{\textsc{Real}}
\newcommand{\Ideal}{\textsc{Ideal}}

\newcommand{\sid}{\ensuremath{\mathsf{sid}}}
\newcommand{\ctx}{\ensuremath{\mathsf{ctx}}}
\newcommand{\pk}{\ensuremath{\mathsf{pk}}}
\newcommand{\Adv}{\mathbf{Adv}}

\renewcommand{\shorttitle}{\scheme: Proactive Threshold VOPRF from Isogenies}

\hypersetup{
  pdftitle={PIVOT: Proactive and Verifiable Threshold Oblivious Pseudorandom Functions from Isogeny Group Actions},
  pdfsubject={Cryptography and Security},
  pdfauthor={},
  pdfkeywords={OPRF, threshold cryptography, proactive security, isogenies, CSIDH, multiparty computation},
  colorlinks=true,
  linkcolor=black,
  citecolor=black,
  urlcolor=black
}

\title{\scheme: Proactive and Verifiable Threshold\\
Oblivious Pseudorandom Functions\\
from Isogeny Group Actions}

\author{
Abhinav Sharma\thanks{This work was carried out as part of a one-year
remote research internship undertaken by Abhinav Sharma when he was doing his Masters at RIE Mysore under the
supervision of Vikas Srivastava.}\\
Indian Institute of Technology Hyderabad\\
Hyderabad, Telangana, India\\
\texttt{abhisha8055@gmail.com}
\And
Vikas Srivastava\\
Department of Mathematics\\
National Institute of Technology Warangal\\
Warangal, Telangana, India\\
\texttt{vikas.math123@gmail.com}
}

\date{}

\begin{document}
\maketitle

\begin{abstract}
Oblivious pseudorandom functions (OPRFs) allow a client to evaluate a keyed pseudorandom function on a private input without revealing that input to the server.  In a threshold OPRF, the secret key is distributed among (n) servers so that any qualified set of at least (t) servers can complete an evaluation, while fewer than (t) shares reveal no information about the key.  Existing isogeny-based threshold OPRFs, however, are primarily designed for static corruption models.  If the same shares remain valid throughout the lifetime of the service, a mobile adversary can compromise different servers over time, accumulate (t) shares from the same sharing state, and eventually recover the master key.

We introduce \scheme{} (\emph{Proactive Isogeny-based Verifiable Oblivious Threshold PRF}), a dealerless threshold VOPRF framework based on effective isogeny group actions.  \scheme{} periodically refreshes the server shares without changing the master key, public key, or previously generated OPRF outputs.  The construction combines Shamir secret sharing, additively homomorphic coefficient commitments, sequential Lagrange-weighted group actions, and joint zero-knowledge relations that link certified shares to their corresponding isogeny actions.  It also supports coordinated epoch transitions, publicly verifiable blame, secure erasure, and committee resharing under a possibly different threshold.

We formalize the functionality of a long-lived proactive threshold VOPRF, prove the correctness of distributed key generation, threshold evaluation, proactive refresh, and committee resharing, and provide a simulation-based security analysis under the vectorization and one-more hidden-group-action assumptions.  As an application, we describe a distributed private lookup service whose encrypted database remains valid across repeated share renewals and committee migrations.
\end{abstract}


\section{Introduction}
\label{sec:introduction}

Cryptographic services that remain in operation for many years face a
different threat model from short-lived protocols.  In a conventional
threshold system, a secret key is divided among several servers so that
no individual server can use the key alone.  This protects against the
failure or compromise of a small number of machines at a fixed point in
time.  It does not, however, prevent an adversary from compromising
different servers at different times.  An attacker may learn the share
held by one server during the first month, the share held by a second
server during a later intrusion, and so on.  If the shares remain
unchanged throughout the lifetime of the service, the attacker can store
each exposed share and eventually collect enough information to
reconstruct the long-lived master key.  The adversary may therefore
break the system without ever controlling the threshold number of
servers simultaneously.

This threat is commonly described through the model of a
\emph{mobile adversary}.  The set of corrupted servers is allowed to
change over time, although the number of servers under the adversary's
control during any individual time period remains below the threshold.
Proactive secret sharing~\citep{herzberg1995} addresses this problem by
periodically replacing the current shares with fresh shares of the same
secret.  Once the old shares and the temporary refresh data have been
securely erased, information obtained during one period cannot simply be
combined with information obtained during a later period.  The secret
itself remains unchanged, but the local representation of that secret is
continually renewed.

The need for such protection is particularly clear in services built
from oblivious pseudorandom functions.  An oblivious pseudorandom
function (OPRF)~\citep{jarecki2009} is an interactive protocol between a
client holding a private input \(x\) and a server holding a secret key
\(k\).  At the end of the interaction, the client learns the value
\(F_k(x)\), while the server learns nothing about \(x\) beyond what is
inherently revealed by the surrounding application.  The client should
not learn the key or obtain useful information about the function on
inputs that were not evaluated through the protocol.  This combination
of input privacy and controlled access to a keyed function makes OPRFs
useful in password-authenticated key exchange, private set
intersection, private keyword search, anonymous credentials, rate
limiting, and encrypted database lookup.

A \emph{verifiable} OPRF strengthens this functionality by allowing the
client to verify that the server evaluated the function under the key
associated with a public commitment or public key.  Verifiability
prevents a malicious server from changing the key from one request to
another, returning an unrelated value, or selectively evaluating under
a key chosen to influence the surrounding application.  A
\emph{threshold} OPRF distributes the secret key among \(n\) servers and
requires the cooperation of at least \(t\) of them to complete an
evaluation.  This removes the monolithic OPRF server as a single point
of compromise and makes the service more tolerant of failures.
Nevertheless, an ordinary threshold OPRF remains vulnerable to gradual
share accumulation if its shares are never refreshed.

The central problem considered in this work is therefore the following,
how can one construct a post-quantum threshold VOPRF that remains secure
for a long period of time, even when the adversary compromises different
servers in different epochs?  A satisfactory construction must preserve
the OPRF key and all previously derived outputs while replacing the
server shares.  It must also ensure that every partial evaluation is
performed with a share certified for the current epoch, that all servers
agree on the active sharing state, and that a malformed contribution can
be detected and attributed.  In addition, the system should support
committee migration, since a genuinely long-lived deployment may need
to replace machines, rotate administrative domains, or change its
threshold parameters without re-encrypting all application data.

\subsection{Why isogeny group actions?}

Post-quantum OPRFs have been studied from several algebraic foundations,
including lattices, codes, and isogenies.  Isogeny group actions are
particularly attractive for this setting because they provide a
commutative action on compact public objects.  Informally, if
\([a]E\) denotes the action of a secret scalar \(a\) on a curve object
\(E\), then the defining composition law
\[
  [a]([b]E)=[a+b]E
\]
resembles the exponent-addition property used in classical
Diffie--Hellman-based constructions.  This algebraic structure is well
suited to threshold evaluation as Lagrange-weighted secret shares can be
applied one after another, and the accumulated action is equal to the
action of the reconstructed secret even though the secret is never
explicitly reconstructed.

The same property also makes client blinding natural.  A client can map
its input to a curve \(X\), apply a fresh random action \([r]\), and send
the blinded value \(B=[r]X\) to the server committee.  After the servers
apply the secret-key action, the client removes the blinding by applying
\([-r]\).  The result is the desired value \([k]X\).  The client-side
randomization hides the original curve from the servers, while the
commutativity of the action permits the server-side contributions to be
combined in any fixed order.

At the same time, isogeny group actions introduce an important
verification challenge.  The secret-sharing layer is naturally
expressed over field elements and polynomial commitments, whereas the
OPRF evaluation takes place in the curve-action domain.  A server must
therefore prove that the scalar used in its group action is exactly the
same scalar that opens its certified share commitment.  The protocol
cannot safely treat these as two unrelated statements.  This
cross-domain witness-consistency requirement plays a central role in the
design of the proof relations used by \scheme{}.

\subsection{Prior work and remaining gap}
\label{sec:prior-work}

The literature on OPRFs, threshold cryptography, proactive sharing, and
isogeny-based protocols has developed along several largely independent
directions.  Classical OPRF constructions and applications were
developed in works such as Jarecki and Liu~\citep{jarecki2009}, and OPRFs
later became a central component of password-authenticated protocols such
as OPAQUE~\citep{opaque}.  These constructions established the
importance of oblivious evaluation but were not designed to provide
post-quantum security.

In the isogeny setting, Heimberger et al.~\citep{heimberger2024}
introduced OPUS, an OT-free Naor--Reingold-style OPRF based on CSIDH,
with security studied in the semi-honest setting.  Delpech de Saint
Guilhem and Pedersen~\citep{dsgp2024} developed proof techniques for
CSIDH arithmetic and constructed a VOPRF with protection against
malicious clients.  Basso~\citep{basso2024} investigated round-efficient
isogeny-based OPRFs.  Levin and Pedersen~\citep{levin2025} developed
faster proof techniques and related verifiable-function constructions.
More recently, Pedersen~\citep{pedersen2026} presented a robust threshold
VOPRF from isogeny group actions, obtaining verifiability and
identifiable aborts through an MPC-among-servers architecture.

Threshold mechanisms for isogeny-based systems were considered earlier
by De Feo and Meyer~\citep{defeo2020}, while Beullens et
al.~\citep{beullens2021} studied distributed key generation for CSIDH.
These works provide important techniques for distributing isogeny-based
secrets, but they do not by themselves solve the long-term share
accumulation problem for a threshold OPRF service.

Proactive secret sharing originates in the work of Herzberg et
al.~\citep{herzberg1995}.  Its core idea is to add a random sharing of
zero to the current sharing polynomial, thereby preserving the secret
while replacing the shares.  Proactive threshold OPRF constructions
have also begun to appear in classical groups.  For example, Baecker et
al.~\citep{baecker2025} proposed a proactive threshold OPRF based on a
one-more gap Diffie--Hellman assumption.  That construction, however, is
not isogeny-based and does not provide the same verifiability mechanisms
considered here.

To the best of our knowledge, the combination required for a long-lived
isogeny-based service has not previously been addressed in one
construction i.e. dealerless threshold key generation, verifiable oblivious
evaluation, proactive renewal against a mobile adversary, secure
transition between epochs, public attribution of malformed
contributions, and committee resharing without changing the OPRF key.
Existing isogeny-based threshold OPRFs generally protect against a
bounded static corruption set.  If the same shares remain valid
indefinitely, a mobile adversary can eventually accumulate a
reconstruction set even though the instantaneous number of corruptions
never reaches \(t\).

\subsection{Overview of \scheme{}}
\label{sec:scheme-overview}

We introduce \scheme{} (\emph{Proactive Isogeny-based Verifiable
Oblivious Threshold PRF}), a protocol framework for long-lived
threshold OPRF evaluation from effective isogeny group actions.
\scheme{} maintains a single master key \(k\), but represents that key
by a different Shamir polynomial in each epoch.  During epoch \(e\), the
servers hold
\[
  s_i^{(e)}=F_e(i),
  \qquad
  F_e(0)=k.
\]
The public coefficient-commitment vector
\(\mathbf A^{(e)}\) certifies the current polynomial, while the public
isogeny key
\[
  \pk=[k]E_0
\]
remains unchanged throughout the lifetime of the system.

The construction separates three forms of state.  The first is the
long-lived semantic state consisting of the master key and public key.
The second is the epoch-specific secret-sharing state consisting of the
current polynomial and server shares.  The third is the public
verification state consisting of coefficient commitments, certificates,
proofs, and signatures.  Proactive refresh changes only the
epoch-specific state.  Committee resharing changes both the committee
and the sharing polynomial.  Neither operation changes the master key,
the public key, or the value of the OPRF on a fixed input.

The protocol comprises five main procedures.  A dealerless distributed
key-generation protocol establishes the initial sharing and public key.
A threshold evaluation protocol allows a client to obtain the OPRF
output from a quorum of servers.  A proactive refresh protocol replaces
the current shares with fresh shares of the same key.  A verification
and blame mechanism identifies malformed contributions and supports
restart with a new quorum.  Finally, a committee-resharing protocol
transfers the same secret to a new committee and may simultaneously
change the threshold.

\subsection{Our contributions}
\label{sec:contributions}

The principal contributions of this work may be summarized as follows.

\begin{enumerate}

  \item We formulate an ideal functionality for a long-lived proactive
        threshold verifiable oblivious pseudorandom function.  The
        functionality, denoted by
        \(\F_{\mathsf{pTVOPRF}}\), captures the complete operational
        life cycle of the service, including dealerless distributed key
        generation, threshold OPRF evaluation, epoch-specific public
        verification state, proactive renewal of secret shares, secure
        erasure of obsolete information, migration to a new server
        committee, and identifiable aborts in the presence of malformed
        protocol contributions.  The corresponding adversarial model
        allows the set of corrupted servers to change from one epoch to
        another, subject to the requirement that the adversary obtains
        fewer than \(t\) valid shares from every individual epoch.  This
        formulation makes explicit the distinction between ordinary
        threshold security, which is usually defined with respect to a
        fixed corruption set, and proactive security, which must remain
        meaningful when compromises occur gradually over the lifetime
        of the system.

  \item We present a dealerless construction that combines an effective
        isogeny group action with Shamir secret sharing and an
        additively homomorphic commitment scheme.  Every server
        contributes independently to the initial distributed
        key-generation polynomial, and the qualified contributions are
        aggregated to define a master key \(k\) that is never
        reconstructed by any individual participant.  The public
        coefficient-commitment vector certifies the active sharing
        polynomial, while the isogeny public key
        \[
          \pk=[k]E_0
        \]
        remains associated with its constant term.  Threshold
        evaluation is performed through a sequential chain of
        Lagrange-weighted group actions on the client's blinded input.
        If \(I\) denotes the selected quorum, then the accumulated
        exponent satisfies
        \[
          \sum_{i\in I}\lambda_i^I s_i^{(e)}=k.
        \]
        Consequently, the final curve produced by the server chain is
        the result of applying the master-key action, even though every
        server uses only its own local share.

  \item We develop a proactive maintenance mechanism that renews the
        server shares without changing the underlying OPRF key.  During
        the transition from epoch \(e\) to epoch \(e+1\), every refresh
        dealer distributes a verifiable sharing of a random polynomial
        \(z_j(X)\) satisfying
        \[
          z_j(0)=0.
        \]
        The next sharing polynomial is defined by
        \[
          F_{e+1}(X)
          =
          F_e(X)+
          \sum_{j\in\mathcal R_e}z_j(X),
        \]
        where \(\mathcal R_e\) denotes the set of qualified refresh dealers,
        and therefore
        \[
          F_{e+1}(0)=F_e(0)=k.
        \]
        The public coefficient commitments are updated
        homomorphically, while a coordinated epoch-transition procedure
        ensures that every successful evaluation is associated with one
        consistent epoch certificate.  We also provide a committee
        resharing procedure in which the Lagrange-weighted shares of an
        old committee are redistributed as the constant terms of fresh
        sharing polynomials for a new committee.  This permits changes
        in committee membership and threshold parameters without
        replacing the master key, changing the public key, or
        invalidating previously generated OPRF outputs.

  \item We identify and formalize the joint NP relations required to
        connect the polynomial-sharing and isogeny-action components of
        the construction.  These relations enforce that a scalar
        committed in the secret-sharing domain is the same scalar used
        in the corresponding group action.  In particular, we define
        the relations
        \[
          \mathcal R_{\mathsf{link}},\qquad
          \mathcal R_{\mathsf{eval}},\qquad
          \mathcal R_{\mathsf{blind}},\qquad
          \mathcal R_{\mathsf{reshare}},
        \]
        which respectively bind DKG constant terms to public-key
        contributions, certify partial threshold evaluations, establish
        correct client blinding, and connect resharing polynomials to
        certified old shares.  We explain why separately proving a
        commitment-opening statement and a group-action statement does
        not automatically establish equality of the witnesses used in
        the two proofs.  The construction therefore requires proofs for
        the complete joint relations, or an explicit witness-equality
        mechanism, rather than an unlinked conjunction of independent
        algebraic statements.

  \item We provide a detailed correctness and simulation-based security
        analysis of the construction and illustrate its use in
        long-lived private lookup services.  The correctness analysis
        establishes that the dealerless DKG creates a valid sharing of
        the master key, that every successful threshold evaluation
        returns the intended value \(F_k(x)\), and that proactive
        refresh and committee resharing preserve the same key \(k\).
        The security analysis is carried out in the stated hybrid model
        under the vectorization and one-more hidden-group-action
        assumptions, together with the hiding and binding properties of
        the commitment scheme, the security of the NIZK and signature
        systems, and the secure-erasure assumption.  As an application,
        we describe a distributed private lookup service whose database
        is indexed or encrypted using OPRF-derived values.  Since the
        refresh and resharing procedures preserve
        \[
          F_k(x)
          =
          \Htwo\!\left(
            \mathsf{PIVOT-out}\|
            \ctx\|
            \pk\|
            x\|
            \mathsf{enc}\bigl([k]\Hone(\ctx\|x)\bigr)
          \right),
        \]
        the database remains valid across repeated share-renewal
        operations and committee migrations.

\end{enumerate}

\subsection{Technical overview}
\label{sec:technical-overview}

We now give an informal description of the construction.  The formal
notation and assumptions appear in Section~\ref{sec:prelim}, and the complete
protocol is specified in Section~\ref{sec:protocol}.

Each server \(S_j\) samples a degree-at-most-\((t-1)\) polynomial
\[
  f_j(X)=\sum_{\ell=0}^{t-1}a_{j,\ell}X^\ell
\]
and distributes its evaluations using verifiable secret sharing.  The
contributions of the qualified dealers (the set $\mathcal Q$ of dealers
that passed VSS verification) are added to obtain
\[
  F_0(X)=\sum_{j\in\mathcal Q}f_j(X).
\]
Server \(S_i\) stores
\(s_i^{(0)}=F_0(i)\), and the master key is
\[
  k=F_0(0)=\sum_{j\in\mathcal Q}a_{j,0}.
\]
The homomorphic coefficient commitments aggregate into a public vector
\(\mathbf A^{(0)}\) that certifies the initial sharing polynomial.

The public key is formed through a sequential chain of constant-term
actions.  Each dealer proves that the exponent used in its contribution
is the same value committed as the constant coefficient of its DKG
polynomial.  The final public curve is therefore
\[
  \pk=[k]E_0.
\]
This link is necessary because a commitment to a scalar and an isogeny
action by that scalar live in different algebraic domains.

For an application context string \(\ctx\), the protocol evaluates the
keyed function
\begin{equation}
  F_k(x)=
  \Htwo\!\left(
    \mathsf{PIVOT-out}\|
    \ctx\|
    \pk\|
    x\|
    \mathsf{enc}\bigl([k]\Hone(\ctx\|x)\bigr)
  \right).
  \label{eq:prf}
\end{equation}
Here \(\Hone\) hashes the input to the group-action orbit,
\(\mathsf{PIVOT-out}\) is a fixed domain-separator string that prevents
cross-protocol hash collisions, \(\ctx\) identifies the OPRF application,
\(\pk=[k]E_0\) is the public key, \(\mathsf{enc}\) is a canonical
byte-encoding of the resulting curve, and \(\Htwo\) derives the final
pseudorandom output.  The epoch number is not
included in the output hash because the intended function must remain
unchanged when the shares are refreshed.

The client computes
\[
  X=\Hone(\ctx\|x),
\]
samples a random scalar \(r\) from the field \(\Zq=\mathbb Z/q\mathbb Z\), and sends
\[
  B=[r]X
\]
to an ordered quorum \(I=\{i_1,\ldots,i_t\}\).  Server \(S_{i_h}\)
computes the next curve in the evaluation chain
\[
  Q_h=
  [\lambda_{i_h}^I s_{i_h}^{(e)}]Q_{h-1},
  \qquad
  Q_0=B,
\]
where \(\lambda_{i_h}^I\) is the Lagrange coefficient for interpolation
at zero.  After all \(t\) actions,
\[
  Q_t=
  \left[
    \sum_{i\in I}\lambda_i^I s_i^{(e)}
  \right]B
  =
  [k]B.
\]
The client removes the blinding,
\[
  Y=[-r]Q_t=[k]X.
\]
It then derives \(F_k(x)\) using Equation~\eqref{eq:prf}.

Every server proves that the share used in its partial action opens the
share commitment derived from the current epoch commitment vector.
The proof statement is bound to the context, session, epoch,
certificate, quorum, and position in the chain.  Consequently, a
partial response cannot be transplanted into an unrelated evaluation or
combined with shares certified under a different epoch state.

At the transition from epoch \(e\) to epoch \(e+1\), each refresh dealer
\(S_j\) samples
\[
  z_j(X)=\sum_{\ell=1}^{t-1}b_{j,\ell}X^\ell.
\]
Because the constant term is zero, adding this polynomial to the current
sharing leaves the secret unchanged.  The next sharing polynomial is
\[
  F_{e+1}(X)
  =
  F_e(X)+
  \sum_{j\in\mathcal R_e}z_j(X),
\]
and hence
\[
  F_{e+1}(0)=F_e(0)=k.
\]
Each server updates its share by adding the refresh evaluations it
receives.  The higher-degree coefficient commitments are updated
homomorphically, while the commitment to the constant term remains
unchanged.

The refresh is completed through a coordinated epoch transition.  The
servers first agree on the next commitment vector and certificate, then
activate the new state.  Evaluation requests are bound to a single
certificate, so shares from \(F_e\) and \(F_{e+1}\) cannot be combined
in one successful chain.  After activation, honest servers erase their
old shares and refresh randomness.  Under the mobile-adversary bound,
the adversary therefore obtains fewer than \(t\) shares from every
individual epoch.

Suppose the old committee uses threshold \(t\) and the new committee
uses threshold \(t'\).  An old quorum \(I\) satisfies
\[
  \sum_{i\in I}\lambda_i^I s_i^{(e)}=k.
\]
Each old server \(S_i\) samples a degree-at-most-\((t'-1)\) polynomial
whose constant term is
\(\lambda_i^I s_i^{(e)}\).  The new servers add the evaluations received
from all old dealers.  Their aggregate polynomial \(G(X)\) satisfies
\[
  G(0)=
  \sum_{i\in I}\lambda_i^I s_i^{(e)}
  =k.
\]
The resharing proof links every old dealer's new constant-term
commitment to its certified old share.  The new committee therefore
obtains a fresh sharing of the same key, while the public key and all
previous OPRF outputs remain unchanged.

\subsection{Architectural comparison with prior threshold VOPRFs}
\label{sec:architectural-comparison}

The recent threshold VOPRF of Pedersen~\citep{pedersen2026} follows an
MPC-among-servers architecture.  From the client's perspective, the
committee behaves like one virtual server, and the protocol can achieve
a transcript whose size is independent of the threshold.  \scheme{}
adopts a different design.  It exposes the threshold structure directly, each server contributes one signed and proven partial group action to a
sequential chain.

This native-threshold architecture has an \(O(t)\) evaluation
transcript, and therefore does not match the constant-size client
transcript of an MPC-emulated virtual server.  Its advantage is that the
secret-sharing state is explicit.  Proactive refresh can be expressed as
the addition of zero-sharing polynomials, and each server's contribution
remains individually attributable.  The construction should therefore
be understood as a different point in the design space rather than a
strict improvement in every performance dimension.

Adding proactive security to an MPC-based threshold VOPRF may require
refreshing the secret-shared MPC state and coordinating the transition of
the virtual server.  In \scheme{}, by contrast, the maintained state is
already represented as Shamir shares with public coefficient
commitments.  Refresh and committee migration are consequently integrated
directly into the protocol architecture.

\subsection{Comparison with related protocols}

Table~\ref{tab:comparison} summarizes the qualitative properties of the
most closely related OPRF and threshold constructions.  The table is
intended to place the protocol in context rather than to provide a full
performance comparison. Concrete efficiency also depends on parameter
selection, the proof system used for the group-action relations, and the
network model.

\begin{table}[ht]
\centering
\caption{Qualitative comparison with selected OPRF and threshold
constructions.  The notation \(\star\) indicates that the referenced
construction does not provide VOPRF-style verifiability,
\(\dagger\) indicates a classical Diffie--Hellman foundation and
\(\ddagger\) indicates a client transcript independent of the threshold
through MPC among the servers.}
\label{tab:comparison}
\small
\begin{tabular}{lcccccc}
\toprule
Protocol
& Threshold
& Oblivious
& Verifiable
& Proactive
& PQ
& Transcript \\
\midrule
Jarecki--Liu~\citep{jarecki2009}
& \(\times\)
& \checkmark
& \(\times\)
& \(\times\)
& \(\times^\dagger\)
& \(O(1)\) \\

OPAQUE~\citep{opaque}
& \(\times\)
& \checkmark
& \(\times\)
& \(\times\)
& \(\times^\dagger\)
& \(O(1)\) \\

OPUS~\citep{heimberger2024}
& \(\times\)
& \checkmark
& \(\times\)
& \(\times\)
& \checkmark
& \(O(1)\) \\

DSGP~\citep{dsgp2024}
& \(\times\)
& \checkmark
& \checkmark
& \(\times\)
& \checkmark
& \(O(1)\) \\

Basso~\citep{basso2024}
& \(\times\)
& \checkmark
& \(\times\)
& \(\times\)
& \checkmark
& \(O(1)\) \\

Baecker et al.~\citep{baecker2025}
& \checkmark
& \checkmark
& \(\times^\star\)
& \checkmark
& \(\times^\dagger\)
& \(O(1)\) \\

De Feo--Meyer~\citep{defeo2020}
& \checkmark
& \(\times\)
& \(\times\)
& \(\times\)
& \checkmark
& -- \\

Pedersen~\citep{pedersen2026}
& \checkmark
& \checkmark
& \checkmark
& \(\times\)
& \checkmark
& \(O(1)^\ddagger\) \\
\midrule

\textbf{\scheme{}}
& \checkmark
& \checkmark
& \checkmark
& \checkmark
& \checkmark
& \(O(t)\) \\
\bottomrule
\end{tabular}
\end{table}

\subsection{Applications}
\label{sec:intro-applications}

The intended use of \scheme{} is not a one-time cryptographic exchange
but a service that must preserve one logical OPRF key over an extended
period.  In a password-authenticated key-exchange deployment, for
example, several authentication servers may jointly provide the OPRF
operation used to protect password records.  Proactive refresh limits
the value of a temporary server compromise without requiring the entire
credential database to be rebuilt.  Since the OPRF key remains
unchanged, records derived under the existing public key remain valid.

A second application is distributed private set intersection or private
membership testing.  A service can encode set elements using OPRF
outputs while distributing the OPRF key across several administrative
domains.  Refresh protects the long-lived key against gradual
compromise, and committee resharing permits a provider to replace or
reorganize the server set without recomputing the encoded database.

The same property is useful in private lookup services.  A provider may
publish an encrypted decision table whose lookup keys are derived from
\(F_k(x)\).  Clients privately evaluate the OPRF and use the result to
recover the matching encrypted entry.  Since neither refresh nor
resharing changes \(F_k\), the table remains valid across system
maintenance operations.  This is particularly important when the table
is large or widely replicated, because re-encryption under a new key
would otherwise be operationally expensive.

Finally, the construction may support anonymous rate-limiting and
credential-checking systems in which stable pseudorandom tags are
required, but no single machine should hold the tagging key.  In such
applications, proactive maintenance provides protection against a
sequence of temporary compromises while preserving stable tags for
legitimate clients.

\subsection{Organization of the paper}

Section~\ref{sec:prelim} introduces the algebraic notation, group-action
model, Shamir secret sharing, homomorphic commitments, NIZK proof
systems, and distributed primitives used throughout the construction.
Section~\ref{sec:protocol} specifies the five sub-protocols of
\scheme: dealerless distributed key generation, threshold evaluation,
proactive share refresh, blame and robust restart, and committee
resharing.  Section~\ref{sec:correctness} establishes correctness of
each sub-protocol and states an overall invariant that is preserved
across the full system lifetime.  Section~\ref{sec:security} provides
the adversarial model, the ideal functionality, and a simulation-based
security proof showing input privacy, key secrecy, and proactive
protection.  Section~\ref{sec:effi} analyses communication and
computational costs, compares the construction with prior work, and
concludes with the properties and limitations of the design.

\section{Preliminaries}
\label{sec:prelim}

This section introduces the algebraic notation and cryptographic tools used
throughout the construction.  We begin with the basic computational and
protocol notation, then describe the effective group-action abstraction on
which the OPRF is built.  We next recall Shamir secret sharing and the
coefficient-wise homomorphic commitments used to certify server shares.
Finally, we summarize the proof systems, communication assumptions, and
distributed primitives required by the protocol.

The presentation is intentionally self-contained.  In particular, we state
the precise algebraic identities used later in the correctness proof and
explain how the public commitment state, the isogeny public key, and the
epoch mechanism are related.  The security assumptions associated with
these objects are stated formally in Section~\ref{sec:security}, the present
section fixes their syntax and the functionality expected from each
primitive.

\subsection{General notation and computational conventions}
\label{sec:notation}

The security parameter is denoted by \(\lambda\in\mathbb N\).  All
algorithms are probabilistic polynomial-time algorithms unless stated
otherwise.  We write \(1^\lambda\) for the unary representation of the
security parameter and use \(\poly(\lambda)\) for an unspecified
polynomial in \(\lambda\).  A function
\(\epsilon:\mathbb N\rightarrow\mathbb R_{\geq 0}\) is
\emph{negligible}, written \(\epsilon(\lambda)=\negl(\lambda)\), if for
every positive polynomial \(p\) there exists \(\lambda_0\) such that
\(\epsilon(\lambda)<1/p(\lambda)\) for every
\(\lambda\geq\lambda_0\).  Two distribution ensembles
\(\{X_\lambda\}_{\lambda\in\mathbb N}\) and
\(\{Y_\lambda\}_{\lambda\in\mathbb N}\) are computationally
indistinguishable, written \(X\approx_c Y\), if no probabilistic
polynomial-time distinguisher separates them with more than negligible
advantage.

For a finite set \(S\), the notation \(x\leftarrow S\) means that \(x\)
is sampled uniformly from \(S\).  More generally,
\(x\leftarrow\mathcal D\) denotes sampling from a distribution
\(\mathcal D\), and \(y\leftarrow\mathsf{Alg}(x)\) denotes the output of
a randomized algorithm.  Concatenation of bit strings is written
\(u\|v\).  The notation \([n]=\{1,\ldots,n\}\) is used for server
indices.  We assume throughout that \(1\leq t\leq n\), where \(n\) is
the number of servers and \(t\) is the reconstruction threshold.

Let \(q\) be a prime and let
\(\Zq=\mathbb Z/q\mathbb Z\).  Server identifiers
\(1,\ldots,n\) are interpreted as distinct nonzero elements of \(\Zq\),
in particular, we require \(n<q\).  This condition ensures that the
denominators occurring in Lagrange interpolation are nonzero and hence
invertible in \(\Zq\).  Unless explicitly stated otherwise, all scalar
addition, subtraction, multiplication, inversion, and polynomial
evaluation are performed in \(\Zq\).  An element of \(\Zq\) written as a
plain integer denotes the unique representative of
that residue class lying in \(\{0,\ldots,q-1\}\).

We use calligraphic letters such as \(\mathcal A\), \(\mathcal S\), and
\(\mathcal F\) for adversaries, simulators, and ideal functionalities,
respectively.  Bold symbols such as
\(\mathbf A=(A_0,\ldots,A_{t-1})\) denote vectors.  A statement--witness
pair for an NP relation \(\mathcal R\) is written
\((\mathsf{stmt};\mathsf{wit})\), with the semicolon separating public
and private data.

\subsection{Protocol identities, sessions, epochs, and certificates}
\label{sec:identifiers}

The protocol is intended to support a long-lived distributed service.
It therefore distinguishes the identity of the application, the identity
of an individual evaluation session, and the epoch to which the active
server shares belong.

The string
\[
  \ctx\in\{0,1\}^*
\]
is an application-specific context string.  It identifies the logical
OPRF instance and serves as a domain separator.  For example, two
applications operated by the same server committee may use distinct
contexts such as \texttt{password-vault-v1} and
\texttt{private-lookup-v1}.  The context is included in the hash-to-orbit
computation and in the final output hash, so an OPRF value derived in one
application cannot be reused as a valid value in another application.

Every evaluation is associated with a session identifier
\[
  \sid\in\{0,1\}^*.
\]
The identifier must be fresh for the relevant protocol instance, and
servers maintain sufficient replay state to reject a second request with
the same identifier.  The session identifier binds together the client
request, the chosen quorum, the sequence of server actions, the NIZK
proofs, and the corresponding signatures.  It therefore prevents a
message produced in one evaluation from being silently inserted into a
different evaluation.

For malicious-client security, freshness alone is not always sufficient
when \(\sid\) is also used to determine the evaluation quorum.  A client
that can try arbitrarily many candidate identifiers may grind over them
until a preferred quorum is selected.  When unbiased quorum selection is
required by the security or liveness analysis, \(\sid\) must therefore
contain an unpredictable contribution from the servers, a public random
beacon, or another source that the client cannot choose adaptively.  The
correctness of the OPRF evaluation does not depend on quorum
unpredictability, it requires only that the selected indices are distinct
and that all participating servers agree on the same quorum.

Time is divided into epochs \(e\in\mathbb N\).  The system begins in
epoch \(0\), immediately after the distributed key-generation procedure,
and advances from epoch \(e\) to epoch \(e+1\) after a successful
proactive refresh.  During epoch \(e\), the servers hold shares
\[
  s_i^{(e)}=F_e(i)
\]
of a polynomial \(F_e(X)\) satisfying \(F_e(0)=k\), where \(k\) is the
long-lived master key.  A refresh replaces \(F_e\) with a new polynomial
\(F_{e+1}\) having the same constant term.

The public state of epoch \(e\) is summarized by an epoch certificate
\[
  \cert_e=
  \bigl(
    e,\ctx,\pk,\mathbf A^{(e)},\mathcal Q_e,
    \mathsf{digest}_e,\ldots
  \bigr),
\]
where \(\mathbf A^{(e)}\) is the vector of commitments to the
coefficients of \(F_e\), \(\mathcal Q_e\) records the active or qualified
server set, and \(\mathsf{digest}_e\) binds the relevant DKG, complaint,
or refresh transcript.  The exact certificate fields depend on the
sub-protocol, but every certificate must bind the epoch number, public
key, and active commitment vector.  An evaluation request includes the
epoch number and a collision-resistant hash of \(\cert_e\).  Each server
checks these values before applying its share.  This check is necessary
because shares from different epochs generally lie on different
polynomials and must never be combined in one interpolation.

We use
\[
  H:\{0,1\}^*\longrightarrow\{0,1\}^{\kappa}
\]
for a collision-resistant hash function used to derive transcript
digests, certificate hashes, and public seeds.  This function is distinct
from the two OPRF-related hash functions \(\Hone\) and \(\Htwo\)
introduced below.

\paragraph{Deterministic quorum selection.}
The public algorithm
\[
  \mathsf{SelectQuorum}(\sid,e,n,t)
\]
returns an ordered \(t\)-element subset
\(I=(i_1,\ldots,i_t)\) of \([n]\).  A convenient implementation derives
a public seed from \(H(\sid\|e\|\ctx)\), expands it into a permutation of
\([n]\), and selects the first \(t\) indices.  The output is ordered
because the sequential evaluation transcript must identify which server
acts at each position.  Every party recomputes the same ordered quorum
and rejects a request containing a different set or ordering.  As noted
above, resistance to quorum grinding additionally requires an
unbiasable source in the derivation of \(\sid\).

\subsection{Execution and communication model}
\label{sec:communication-model}

The protocol is described in a hybrid model that provides authenticated
point-to-point channels, authenticated broadcast, a public-key
infrastructure for signatures, a commitment functionality, a NIZK
functionality, and secure erasure.  These abstractions isolate the main
group-action and proactive-sharing ideas from lower-level network
mechanisms.

An authenticated private channel guarantees that the receiver learns the
identity of the sender and that the message is neither modified nor read
by parties other than the designated sender and receiver.  Such channels
are used for distributing VSS shares and refresh contributions.
Authenticated broadcast guarantees that all honest parties receive the
same message attributed to the same sender.  It is used for coefficient
commitments, complaints, qualification decisions, epoch certificates,
and transition messages.  Digital signatures make protocol evidence
portable, a signed partial evaluation or refresh contribution can later
be included in a publicly verifiable blame certificate.

The protocol does not require all honest servers to change epochs at the
same physical instant.  Instead, correctness is formulated through
\emph{epoch-consistent evaluation}, every successful evaluation must be
bound to one epoch certificate, and every server in the chain must act
using the share certified by that certificate.  If servers temporarily
disagree about the active epoch during a transition, a request that
crosses the boundary is rejected or retried rather than completed using
a mixture of old and new shares.

\subsection{Effective isogeny group actions}
\label{sec:group-action}

The OPRF is built from an efficiently computable action of a finite
abelian group on a set of elliptic-curve objects.  We first present the
abstract structure needed by the protocol and then explain its relation
to CSIDH-style actions.

Let \(G=\langle g\rangle\) be a cyclic group of prime order \(q\), and
let \(\E\) be a finite set.  An action of \(G\) on \(\E\) is a map
\[
  \star:G\times\E\longrightarrow\E
\]
satisfying
\[
  1_G\star E=E
  \qquad\text{and}\qquad
  (uv)\star E=u\star(v\star E)
\]
for all \(u,v\in G\) and \(E\in\E\).  We write
\[
  [a]E:=g^a\star E
  \qquad\text{for }a\in\Zq.
\]
Under this notation, the action law becomes
\begin{equation}
  [0]E=E,
  \qquad
  [a]([b]E)=[a+b]E.
  \label{eq:action-law}
\end{equation}
The inverse action is represented by \([-a]E\), and therefore
\[
  [-a]([a]E)=E.
\]
These identities are the algebraic reason that the client can blind and
later unblind an input and that the servers can accumulate
Lagrange-weighted shares through a sequential chain of actions.

The action is \emph{free} if
\[
  [a]E=E
  \quad\Longrightarrow\quad
  a=0\pmod q,
\]
and it is \emph{transitive} if, for every \(E,E'\in\E\), there exists
\(a\in\Zq\) such that \(E'=[a]E\).  A free and transitive action makes
\(\E\) a principal homogeneous space, or torsor, for \(G\).  In
particular, for every ordered pair \((E,E')\in\E^2\), there is a unique
scalar \(a\in\Zq\) satisfying \(E'=[a]E\).

In the intended isogeny setting, \(\E\) consists of canonical
representatives of isomorphism classes of supersingular elliptic curves
within one class-group orbit, and the action is induced by ideal classes.
The protocol requires the following operations to be efficient,

\begin{enumerate}[nosep,label=(\roman*)]
  \item Testing whether an encoded object is a valid member of \(\E\),
  \item Computing \([a]E\) for \(a\in\Zq\) and \(E\in\E\),
  \item Computing inverse actions \([-a]E\),
  \item Canonically encoding the resulting curve object.
\end{enumerate}

The associated \emph{vectorization problem} is the following, given
\(E,E'\in\E\) with \(E'=[a]E\), recover \(a\).  A hard homogeneous space
is, informally, a free and transitive action for which the action is
efficient but vectorization is computationally infeasible.  The formal
vectorization assumption used by \scheme{} is stated in
Section~\ref{sec:security}.

\begin{remark}[Prime-order abstraction and CSIDH]
\label{rem:prime-order-abstraction}
The scalar notation above assumes a publicly specified cyclic
prime-order action.  This abstraction is convenient because Shamir
sharing, Lagrange interpolation, and all proof witnesses then live in the
same field \(\Zq\).  The full ideal class group used in standard CSIDH
descriptions is generally an arbitrary finite abelian group rather than
a known cyclic group of prime order.  A concrete instantiation must
therefore either identify a suitable prime-order cyclic subgroup with an
efficiently computable action or generalize the sharing layer from field
scalars to an appropriate product or module representation of the class
group.  The protocol and proofs in this paper apply directly to the
prime-order abstraction stated above, adapting them to a full
vector-exponent CSIDH representation requires an explicit parameter and
encoding treatment.
\end{remark}

\subsection{Hashing to the action set and canonical encodings}
\label{sec:hashing}

We fix a public base element \(E_0\in\E\).  The long-lived public key is
\[
  \pk=[k]E_0,
\]
where \(k\in\Zq\) is the master secret shared among the servers.

The function
\[
  \Hone:\{0,1\}^*\longrightarrow\E
\]
maps arbitrary strings to valid elements of the action set.  We model
\(\Hone\) as a random oracle into \(\E\), or equivalently assume a
hash-to-orbit procedure whose output distribution is computationally
indistinguishable from uniform over the relevant orbit.  The context is
included in every call,
\[
  X=\Hone(\ctx\|x).
\]
This domain separation prevents the same raw input from being interpreted
as the same OPRF point under unrelated applications.

The function
\[
  \mathsf{enc}:\E\longrightarrow\{0,1\}^*
\]
is a canonical injective encoding.  Canonicality is important because
elliptic curves may admit several mathematically equivalent
representations.  The protocol hashes an encoding of the resulting curve
object, so two parties acting on the same element of \(\E\) must obtain
the same byte string.  The injectivity requirement is with respect to the
canonical representatives used by the protocol.

The output hash
\[
  \Htwo:\{0,1\}^*\longrightarrow\{0,1\}^{\ell}
\]
is modelled as a random oracle and is domain-separated from \(\Hone\).
For a fixed public string \(\mathsf{PIVOT-out}\), the keyed
function evaluated by the protocol is
\begin{equation}
  F_k(x)=
  \Htwo\!\left(
    \mathsf{PIVOT-out}\|
    \ctx\|
    \pk\|
    x\|
    \mathsf{enc}\bigl([k]\Hone(\ctx\|x)\bigr)
  \right).
  \label{eq:prf-prelim}
\end{equation}
Including \(\ctx\), \(\pk\), and \(x\) in the final hash binds the output
to the application, the public key, and the exact client input.  The
epoch number is deliberately omitted, proactive refresh changes the
sharing polynomial but not \(k\), and therefore the OPRF output should
remain stable across epochs.

\subsection{Shamir secret sharing}
\label{sec:shamir}

Shamir secret sharing distributes a field element among \(n\) parties so
that any \(t\) shares reconstruct the secret, whereas fewer than \(t\)
shares reveal no information about it.

To share a secret \(k\in\Zq\), sample coefficients
\(a_1,\ldots,a_{t-1}\leftarrow\Zq\) and define
\[
  f(X)=k+\sum_{\ell=1}^{t-1}a_\ell X^\ell.
\]
The share of server \(S_i\) is
\[
  s_i=f(i).
\]
The polynomial has degree at most \(t-1\), its degree may be smaller if
one or more leading coefficients happen to be zero.

Let \(I\subseteq[n]\) contain at least \(t\) distinct indices.  For each
\(i\in I\), define the Lagrange coefficient for interpolation at zero by
\begin{equation}
  \lambda_i^I
  =
  \prod_{\substack{j\in I\\j\neq i}}
  \frac{-j}{i-j}
  \pmod q.
  \label{eq:lagrange-coefficient}
\end{equation}
Then
\begin{equation}
  f(0)
  =
  \sum_{i\in I}\lambda_i^I f(i)
  =
  \sum_{i\in I}\lambda_i^I s_i.
  \label{eq:lagrange}
\end{equation}
In the evaluation protocol, \(I\) has exactly \(t\) elements, but the
identity holds for any larger set as well when the coefficients are
computed for that set.

The privacy of Shamir sharing is information-theoretic.  For any set
\(J\subseteq[n]\) with \(|J|<t\), the joint distribution of
\(\{f(j)\}_{j\in J}\) is uniform over \(\Zq^{|J|}\) for every fixed
secret \(k\).  Equivalently, the observed shares can be extended to a
degree-at-most-\((t-1)\) polynomial having any desired constant term.
This property is the basis for threshold key privacy and for the
simulation of corrupted servers.

The following consequence connects Shamir interpolation to the group
action used by the OPRF.

\begin{lemma}[Interpolation through the group action]
\label{lem:interpolation-action}
Let \(s_i=f(i)\) be Shamir shares of \(k=f(0)\), and let
\(I\subseteq[n]\) be a reconstruction set.  Then, for every
\(E\in\E\),
\[
  \left[
    \sum_{i\in I}\lambda_i^I s_i
  \right]E
  =
  [k]E.
\]
Moreover, if the actions are applied sequentially in any order, then
\[
  [\lambda_{i_t}^I s_{i_t}]
  \cdots
  [\lambda_{i_2}^I s_{i_2}]
  [\lambda_{i_1}^I s_{i_1}]E
  =
  [k]E.
\]
\end{lemma}

\begin{proof}
Equation~\eqref{eq:lagrange} gives
\(\sum_{i\in I}\lambda_i^I s_i=k\).  Repeated application of
Equation~\eqref{eq:action-law} shows that sequential actions add their
exponents.  The claimed identities follow immediately.
\end{proof}

\paragraph{Zero-sharing polynomials.}
A polynomial
\[
  z(X)=\sum_{\ell=1}^{t-1}b_\ell X^\ell
\]
is called a zero-sharing polynomial because \(z(0)=0\).  If \(f\)
shares \(k\), then
\[
  f'(X)=f(X)+z(X)
\]
also shares \(k\), since \(f'(0)=f(0)\).  The refreshed shares are
\(s_i'=s_i+z(i)\).  This elementary identity is the algebraic foundation
of proactive refresh.

\subsection{Additively homomorphic coefficient commitments}
\label{sec:commitment}

The protocol uses commitments to certify the coefficients of sharing
polynomials without publishing those coefficients.  We stress that this
is a vector of ordinary commitments to coefficients, not a succinct
polynomial-commitment scheme in the sense of KZG-style commitments.

A commitment scheme consists of algorithms
\[
  \Com:\Zq\times\R\longrightarrow\C
  \qquad\text{and}\qquad
  \mathsf{Open},
\]
where \(\R\) is the randomness space and \(\C\) is the commitment space.
We write \(\Com(a;\rho)\) for a commitment to \(a\in\Zq\) using
randomness \(\rho\in\R\).  The scheme must satisfy the following
properties.

\begin{enumerate}[label=(\roman*)]
  \item \textbf{Correctness.}
        An honestly generated commitment opens successfully to the
        committed value and randomness.

  \item \textbf{Hiding.}
        Commitments to any two values are computationally
        indistinguishable, or statistically indistinguishable when a
        statistically hiding instantiation is used.

  \item \textbf{Binding.}
        No probabilistic polynomial-time adversary can produce one
        commitment together with valid openings to two distinct field
        elements, except with negligible probability.

  \item \textbf{Additive homomorphism.}
        There is a public operation \(\oplus\) on commitments such that
        \[
          \Com(a;\rho)\oplus\Com(b;\eta)
          =
          \Com(a+b;\rho+\eta).
        \]

  \item \textbf{Public scalar multiplication.}
        For every public \(c\in\Zq\), there is a public operation
        \(\odot\) satisfying
        \[
          c\odot\Com(a;\rho)
          =
          \Com(ca;c\rho).
        \]
\end{enumerate}

Post-quantum instantiations may be obtained from standard lattice-based
commitment techniques, provided that the chosen scheme supports the
required additive operations over the scalar domain used by the sharing
scheme.  The security proof treats the commitment layer modularly and
uses only the properties listed above.

Let
\[
  f(X)=\sum_{\ell=0}^{t-1}a_\ell X^\ell
\]
and let
\[
  A_\ell=\Com(a_\ell;\rho_\ell)
  \qquad\text{for }0\leq\ell<t.
\]
The vector
\[
  \mathbf A=(A_0,\ldots,A_{t-1})
\]
is called the coefficient-commitment vector of \(f\).  For a public
index \(i\in[n]\), define
\begin{equation}
  \EvalCom(\mathbf A,i)
  :=
  \bigoplus_{\ell=0}^{t-1}
  i^\ell\odot A_\ell.
  \label{eq:evalcom-definition}
\end{equation}
By homomorphism,
\begin{align}
  \EvalCom(\mathbf A,i)
  &=
  \Com\!\left(
    \sum_{\ell=0}^{t-1}a_\ell i^\ell;
    \sum_{\ell=0}^{t-1}\rho_\ell i^\ell
  \right) \notag\\
  &=
  \Com(f(i);\omega_i),
  \label{eq:evalcom}
\end{align}
where
\[
  \omega_i=
  \sum_{\ell=0}^{t-1}\rho_\ell i^\ell.
\]
Thus anyone can derive a commitment to the share that server \(S_i\)
should hold, while only \(S_i\) needs to know the opening
\((f(i),\omega_i)\).

This mechanism is used repeatedly.  During the DKG, each dealer
commits to its polynomial and receivers verify the shares they obtain.
After aggregation, the public coefficient commitments add to a
commitment vector for the aggregate sharing polynomial.  During
proactive refresh, commitments to zero-polynomial coefficients are added
to the active vector.  During evaluation, a server proves that the share
used in its group action opens the publicly derived commitment
\(\EvalCom(\mathbf A^{(e)},i)\).

\subsection{NP relations and witness consistency}
\label{sec:np-relations}

The proof obligations appearing in \scheme{} are expressed as
\emph{NP relations}.  Introducing this abstraction explicitly is useful
because the protocol does not merely require a party to prove knowledge
of an isolated secret.  Instead, a prover must often demonstrate that the
same private value is consistent with several public objects produced in
different algebraic domains.

An NP relation is a polynomial-time decidable relation
\[
  \mathcal R\subseteq\mathcal X\times\mathcal W,
\]
where \(\mathcal X\) is the statement space and \(\mathcal W\) is the
witness space.  A pair \((x,w)\in\mathcal R\) means that the public
statement \(x\) is true with respect to the private witness \(w\).
Membership can be checked efficiently by a deterministic polynomial-time
verification algorithm
\[
  \mathsf{Check}_{\mathcal R}(x,w)\in\{0,1\}.
\]
The language associated with \(\mathcal R\) is
\[
  L_{\mathcal R}
  =
  \bigl\{
    x\in\mathcal X:
    \exists\,w\in\mathcal W
    \text{ such that }(x,w)\in\mathcal R
  \bigr\}.
\]
A zero-knowledge proof convinces the verifier that
\(x\in L_{\mathcal R}\) without revealing the witness \(w\).

In the present construction, the public statement normally contains
protocol metadata together with algebraic objects such as commitments,
curve encodings, server indices, Lagrange coefficients, epoch
certificates, and input/output curves.  The witness contains private
field elements and commitment randomness, such as a Shamir share
\(s_i^{(e)}\), an opening \(\omega_i^{(e)}\), a blinding scalar \(r\),
or the constant coefficient of a resharing polynomial.

It is important that every relation include all public values needed to
identify the precise protocol execution in which the proof is valid.
Accordingly, a concrete implementation should bind the statement to the
protocol name, relation identifier, context \(\ctx\), session identifier
\(\sid\), epoch \(e\), certificate hash, server identity, quorum, and
position in the evaluation chain whenever these values are relevant.
This prevents a proof generated for one session, epoch, or sub-protocol
from being replayed as a valid proof in another.

\paragraph{Simple and joint relations.}
Some NP relations assert one algebraic property.  For example, a
commitment-opening relation may be written as
\[
  \mathcal R_{\mathsf{open}}
  =
  \bigl\{
    (C;\,a,\rho):
    C=\Com(a;\rho)
  \bigr\}.
\]
Similarly, an action relation may assert
\[
  \mathcal R_{\mathsf{act}}
  =
  \bigl\{
    (Q,Q';\,a):
    Q'=[a]Q
  \bigr\}.
\]
The relations required by \scheme{}, however, are generally
\emph{joint relations}.  A joint relation requires one witness to satisfy
several conditions simultaneously.  A representative example is
\[
  \mathcal R_{\mathsf{joint}}
  =
  \bigl\{
    (C,Q,Q';\,a,\rho):
    C=\Com(a;\rho)
    \ \wedge\
    Q'=[a]Q
  \bigr\}.
\]
The significance of this formulation is that the scalar opening the
commitment must be exactly the scalar used in the group action.  It is
not enough to prove separately that \(C\) opens to some value \(a\) and
that \(Q'\) is obtained from \(Q\) using some possibly different value
\(a'\).  The correctness and security of the protocol depend on
witness equality across the two conditions.

\paragraph{Relations used by \scheme{}.}
The protocol later defines four concrete relations.

\begin{enumerate}[label=(\roman*)]
  \item The \emph{link relation}
        \(\mathcal R_{\mathsf{link}}\) binds a DKG dealer's committed
        constant coefficient to the exponent used in its public-key
        contribution.

  \item The \emph{evaluation relation}
        \(\mathcal R_{\mathsf{eval}}\) proves that a server's partial
        group action was computed with the same Shamir share that opens
        the share commitment derived from the active epoch commitment
        vector.

  \item The \emph{blinding relation}
        \(\mathcal R_{\mathsf{blind}}\) proves that a client request is a
        correctly blinded hash-derived input rather than an arbitrary
        curve chosen independently of an input.

  \item The \emph{resharing relation}
        \(\mathcal R_{\mathsf{reshare}}\) binds the constant term of an
        old server's resharing polynomial to its certified
        Lagrange-weighted share.
\end{enumerate}

They are logically distinct and must
be domain-separated at the proof-system level.  In particular, an
accepting proof for \(\mathcal R_{\mathsf{eval}}\) must never be
interpretable as a proof for \(\mathcal R_{\mathsf{link}}\), even when
some of their public inputs have similar encodings.

\paragraph{Relation completeness and soundness.}
For every honest protocol execution, the prescribed witness must satisfy
the corresponding relation.  This is the relation-level correctness
condition underlying NIZK completeness.  Conversely, if a statement is
false, for example, if a server applies an action using a scalar
different from its certified share then no valid witness should exist.
NIZK soundness ensures that such a false statement cannot be accepted
except with negligible probability.

\paragraph{Relation efficiency.}
An NP relation is useful only when its verification predicate can be
evaluated efficiently.  Commitment equations are usually inexpensive,
whereas the predicate
\[
  Q'=[a]Q
\]
may require evaluating the complete isogeny group action inside the
proof system.  The abstract protocol therefore treats proofs for these
relations as modular cryptographic building blocks.  A concrete
implementation must specify how the action predicate is represented and
must account for its proof-generation and verification costs.

The distinction between defining a relation and instantiating a proof
system for that relation is important.  The protocol and its security
proof require the four relations to be well defined and efficiently
decidable.  They do not, by themselves, imply that a practically
efficient circuit or proof system for the isogeny-action predicate is
already available.

\subsection{Non-interactive zero-knowledge proofs}
\label{sec:nizk}

Several protocol steps require a party to prove that one private scalar
is used consistently in two different algebraic domains.  For example,
an evaluation server must show both that \(s_i^{(e)}\) opens its
certified share commitment and that the same scalar determines the
action applied to the incoming curve.  These statements are expressed as
NP relations and proved using non-interactive zero-knowledge proofs.

For an NP relation
\[
  \mathcal R\subseteq\mathcal X\times\mathcal W,
\]
a common-reference-string NIZK system consists of algorithms
\[
  \bigl(
    \NIZK.\Setup,
    \NIZK.\Prove,
    \NIZK.\Verify
  \bigr).
\]
The setup algorithm generates a common reference string
\(\mathsf{crs}\).  Given a statement \(x\in\mathcal X\) and a witness
\(w\in\mathcal W\) satisfying \((x,w)\in\mathcal R\), the prover
computes
\[
  \pi\leftarrow
  \NIZK.\Prove_{\mathcal R}(\mathsf{crs},x;w).
\]
The verifier outputs
\[
  b\leftarrow
  \NIZK.\Verify_{\mathcal R}(\mathsf{crs},x,\pi).
\]

We require the following properties.

\begin{enumerate}[label=(\roman*)]
  \item \textbf{Completeness.}
        An honestly generated proof for a true statement is accepted,
        except with negligible probability.  Perfect completeness may
        be assumed when provided by the selected proof system.

  \item \textbf{Soundness.}
        No probabilistic polynomial-time prover can produce an accepting
        proof for a false statement, except with negligible probability.

  \item \textbf{Zero knowledge.}
        There exists a simulator \(\NIZK.\Sim\) that, given an
        appropriate simulation trapdoor and a statement, produces a
        proof computationally indistinguishable from a real proof without
        knowing a witness.

  \item \textbf{Simulation extractability.}
        There exists an extractor \(\NIZK.\Ext\) such that, even after
        observing simulated proofs, an adversary that produces a fresh
        accepting proof for a new statement yields a corresponding valid
        witness, except with negligible probability.
\end{enumerate}

The full malicious-security analysis uses simulation extractability to
recover the witnesses underlying adversarially generated proofs.
A semi-honest analysis does not require this extraction property because
all parties are assumed to follow the prescribed algorithms.  We retain
the stronger primitive in the protocol specification because \scheme{}
also aims to provide public verifiability and identifiable
misbehaviour.

We use the notation
\[
  \NIZK.\Prove_{\mathcal R}(\mathsf{stmt};\mathsf{wit})
\]
and
\[
  \NIZK.\Verify_{\mathcal R}(\mathsf{stmt},\pi)
\]
when the common reference string is understood.  The simulator and
extractor are written \(\NIZK.\Sim\) and \(\NIZK.\Ext\), respectively.

The protocol uses four domain-separated proof relations,
\(\mathcal R_{\mathsf{link}}\),
\(\mathcal R_{\mathsf{eval}}\),
\(\mathcal R_{\mathsf{blind}}\), and
\(\mathcal R_{\mathsf{reshare}}\).   Domain separation means that a proof
created for one relation cannot be interpreted as a proof for another
relation, even if their public inputs have similar encodings.

\begin{remark}[Joint-relation requirement]
The relations used by \scheme{} are not merely independent conjunctions
of a commitment-opening statement and a group-action statement.  They
require the \emph{same} scalar witness to satisfy both components.
Accordingly, a concrete NIZK instantiation must prove the complete joint
relation and enforce equality of the witness across the commitment and
isogeny domains.  Independent proofs of the two components are
insufficient unless they are connected by an explicit witness-equality
mechanism.
\end{remark}

\subsection{Verifiable secret sharing}
\label{sec:vss}

A verifiable secret-sharing protocol allows a dealer to distribute
evaluations of a polynomial while enabling each receiver to verify that
its share is consistent with one common committed polynomial.  VSS is
required because ordinary Shamir sharing offers no mechanism for
detecting a dealer that sends unrelated values to different receivers.

In the commitment-based VSS used by \scheme{}, a dealer samples
\[
  f(X)=\sum_{\ell=0}^{t-1}a_\ell X^\ell
\]
and publishes the coefficient commitments
\[
  A_\ell=\Com(a_\ell;\rho_\ell).
\]
For receiver \(S_i\), it computes
\[
  u_i=f(i),
  \qquad
  \varrho_i=\sum_{\ell=0}^{t-1}i^\ell\rho_\ell
\]
and sends \((u_i,\varrho_i)\) through an authenticated private channel.
The receiver checks
\begin{equation}
  \Com(u_i;\varrho_i)
  =
  \EvalCom(\mathbf A,i).
  \label{eq:vss-verification}
\end{equation}
If the equation holds, the share is consistent with the committed
coefficients.  If it fails, the receiver issues a complaint according to
the complaint-resolution procedure of the protocol.

Complaint resolution must satisfy two goals.  First, all honest parties
must reach the same decision about whether the dealer remains qualified.
Second, the resolution transcript must not disclose enough valid shares
to violate the threshold privacy guarantee.  The protocol records a
collision-resistant digest of the resolution transcript in the relevant
certificate so that later parties can verify which dealers were included
in the aggregate state.

We write \(\mathcal Q\) for the set of dealers that remain qualified
after VSS verification and complaint resolution.  A dealer is qualified
if and only if its shares are accepted by at least \(t\) honest receivers
after complaint resolution.  In an all-honest
execution, every dealer is qualified.

\subsection{Dealerless distributed key generation}
\label{sec:dkg-prelim}

A distributed key-generation protocol produces a sharing of a randomly
generated key without appointing a trusted dealer.  In \scheme{}, every
server acts as a VSS dealer.  Server \(S_j\) samples a polynomial
\(f_j(X)\), distributes its evaluations, and publishes commitments to its
coefficients.  After the qualification phase, each server adds all
accepted contributions.  If fewer than \(t\) dealers qualify, the DKG
aborts and is restarted with a fresh session identifier, while the public
state remains uninitialised until a successful run completes.
Otherwise, the qualified set \(\mathcal Q\) is nonempty and each server
computes
\[
  s_i=
  \sum_{j\in\mathcal Q}f_j(i).
\]
The resulting aggregate polynomial is
\[
  F(X)=
  \sum_{j\in\mathcal Q}f_j(X),
\]
and its constant term
\[
  k=F(0)=
  \sum_{j\in\mathcal Q}f_j(0)
\]
is the distributed master key.  No server needs to reconstruct this
sum.  The coefficient commitments aggregate homomorphically,
\[
  A_\ell=
  \bigoplus_{j\in\mathcal Q}A_{j,\ell}.
\]

The DKG additionally constructs the public key
\(\pk=[k]E_0\).  Since the commitment space and the group-action space
are distinct algebraic domains, the protocol uses link proofs to show
that each public-key contribution is formed with the same constant
coefficient committed by the corresponding dealer.  This link is what
binds the Shamir sharing to the public isogeny key.

\subsection{Digital signatures and public-key infrastructure}
\label{sec:signatures}

Each server \(S_i\) holds a post-quantum digital-signature key pair
\[
  (\mathsf{sk}_i^{\mathsf{sig}},
   \mathsf{vk}_i^{\mathsf{sig}}).
\]
Verification keys are authenticated through a public-key infrastructure
and are known to all protocol participants.  We require correctness and
existential unforgeability under chosen-message attack
(EUF--CMA).

Signatures serve three purposes.  They authenticate partial evaluation
messages and refresh contributions, prevent one server from being framed
for a message generated by another party, and make blame evidence
transferable to third parties.  Every signed message includes sufficient
domain separation, including the protocol name, context, session
identifier, epoch, message type, and relevant public values.  This
prevents a signature issued in one sub-protocol from being replayed as a
valid signature in another.

\subsection{Secure erasures and proactive security}
\label{sec:erasures}

Proactive security is meaningful only if obsolete local state can be
removed.  After a successful transition from epoch \(e\) to epoch
\(e+1\), an honest server erases its old share \(s_i^{(e)}\), the
corresponding commitment-opening randomness, received refresh
contributions, temporary VSS randomness, and any other state from which
the old share could be reconstructed.

We model erasure through an ideal functionality
\(\F_{\mathsf{erase}}\).  Once a value has been erased, a later
corruption of the server does not reveal that value.  This assumption is
standard in proactive secret-sharing models.  Without secure erasure, a
mobile adversary could compromise one server in each epoch and recover
all historical shares stored on disk, eventually collecting enough
same-epoch information to reconstruct the master key.

Erasure is invoked only after the next epoch state has been accepted.
Erasing the old state too early could destroy availability if the epoch
transition later aborts.  Conversely, retaining the old state after
activation weakens proactive security.  The transition protocol
therefore separates preparation of the next shares from their activation
and performs erasure only after the next certificate has been committed.

\subsection{Summary of the maintained public and private state}
\label{sec:state-summary}

At the beginning of an active epoch \(e\), server \(S_i\) holds the
private state
\[
  \mathsf{st}_i^{(e)}
  =
  \bigl(
    s_i^{(e)},
    \omega_i^{(e)},
    \mathsf{sk}_i^{\mathsf{sig}},
    \mathsf{ReplayState}_i
  \bigr),
\]
while the public state contains
\[
  \mathsf{pst}^{(e)}
  =
  \bigl(
    \ctx,e,\pk,\mathbf A^{(e)},\cert_e,
    \{\mathsf{vk}_i^{\mathsf{sig}}\}_{i=1}^{n}
  \bigr).
\]
The fundamental state invariant is that there exists a polynomial
\[
  F_e(X)=
  \sum_{\ell=0}^{t-1}a_\ell^{(e)}X^\ell
\]
such that
\[
  F_e(0)=k,
  \qquad
  s_i^{(e)}=F_e(i),
  \qquad
  A_\ell^{(e)}
  =
  \Com(a_\ell^{(e)};\rho_\ell^{(e)}),
  \qquad
  \pk=[k]E_0.
\]
The DKG establishes this invariant, threshold evaluation uses it,
proactive refresh preserves it while changing the nonconstant
coefficients, and committee resharing transfers it to a new committee.
The correctness analysis in Section~\ref{sec:correctness} proves these
claims formally.

\section{The \scheme\ Protocol}
\label{sec:protocol}

This section presents the complete \scheme{} construction.  Before giving
the individual algorithms, it is useful to describe how the different
components fit together.  The protocol maintains a long-lived master key
\(k\in\Zq\), but no server stores \(k\) directly.  Instead, during every
epoch \(e\), the servers hold evaluations of a degree-at-most-\((t-1)\)
Shamir polynomial
\[
  F_e(X)\in\Zq[X]
  \qquad\text{with}\qquad
  F_e(0)=k.
\]
Server \(S_i\) stores the share \(s_i^{(e)}=F_e(i)\).  The public
coefficient-commitment vector \(\mathbf A^{(e)}\) certifies the polynomial
that defines the current shares, while the public key
\(\pk=[k]E_0\) commits to the unchanged master secret in the isogeny
group-action domain.

The protocol can be viewed as a sequence of state-creation, state-use,
and state-maintenance procedures.  The dealerless DKG creates the initial
sharing and public key.  The threshold evaluation protocol uses any
qualified set of \(t\) current shares to evaluate the OPRF without
reconstructing \(k\).  Proactive refresh replaces \(F_e\) by a newly
randomized polynomial \(F_{e+1}\) having the same constant term, thereby
invalidating previously exposed shares.  The verification and blame
mechanism makes malformed contributions publicly attributable.  Finally,
committee resharing transfers the same master key to a new collection of
servers, possibly under a different threshold.

Three forms of consistency are maintained throughout the construction.
First, the shares held by the servers must be evaluations of the polynomial
represented by the public commitment vector.  Second, the constant
coefficient of that polynomial must correspond to the exponent used to
form \(\pk\).  Third, all servers participating in one evaluation must use
the same epoch certificate.  The NIZK relations enforce the first two forms of
consistency, while explicit epoch and session identifiers enforce the
third.

We describe \scheme{} in the
\((\F_{\mathsf{auth}},\F_{\mathsf{broadcast}},\F_{\mathsf{com}},
\F_{\mathsf{nizk}},\F_{\mathsf{sig}},\F_{\mathsf{erase}})\)-hybrid
model, following the universal-composability framework of
Canetti~\citep{canetti2001}.  \(\F_{\mathsf{auth}}\) provides
authenticated private channels for share delivery,
\(\F_{\mathsf{broadcast}}\) guarantees that every honest party receives
identical copies of public announcements such as commitments, complaints,
and certificates, \(\F_{\mathsf{com}}\) realizes the homomorphic
commitment scheme of \Cref{sec:commitment}, \(\F_{\mathsf{nizk}}\)
realizes the simulation-extractable NIZK of \Cref{sec:nizk}, and
\(\F_{\mathsf{sig}}\) provides EUF-CMA signatures.  Secure erasure is
handled by \(\F_{\mathsf{erase}}\) and is invoked after a successful
epoch transition.

The setup phase fixes the algebraic environment and the auxiliary
cryptographic mechanisms used by every later sub-protocol.  These
parameters are system-wide rather than epoch-specific.  In particular,
the group-action parameters, base curve, hash functions, commitment
scheme, and NIZK reference string remain unchanged when the shares are
refreshed.  This stability is important because proactive maintenance is
intended to protect a long-lived service without changing its public key
or invalidating previously derived OPRF outputs.

The setup also establishes the authenticated identities of the servers.
Signatures do not hide any data, their role is to bind each public
message to its sender and to make later blame certificates independently
verifiable.  The epoch counter begins at zero and is advanced only after
the refresh transition has completed.

\begin{enumerate}[nosep,leftmargin=*,label=\textbf{P\arabic*.}]
  \item A prime $q$ and base curve $E_0\in\E$ are fixed for the target
        security level.
  \item Hash functions $\Hone:\{0,1\}^*\to\E$ and
        $\Htwo:\{0,1\}^*\to\{0,1\}^{\ell}$ are fixed (random oracles).
  \item An additively homomorphic commitment scheme $\Com:\Zq\times\R\to\C$
        is fixed.
  \item A CRS for the four SE-NIZK relations is generated and published.
  \item Each server $S_i$ generates a signing keypair and verification keys are distributed via a PKI.
  \item The epoch counter $e\leftarrow 0$ is initialised.
\end{enumerate}

\subsection{Protocol 1: Dealerless distributed key generation}
\label{sec:prot1}

The first task is to create the master key without appointing a trusted
dealer.  A conventional trusted-dealer construction would sample a
polynomial \(F_0\), distribute \(F_0(i)\) to server \(S_i\), and publish
the corresponding verification information.  Such a dealer would,
however, learn the complete master key and become a permanent point of
trust.

\scheme{} avoids this problem by allowing every server to
contribute an independently sampled polynomial.  The sum of all
qualified contributions becomes the initial sharing polynomial.

More precisely, each server \(S_j\) acts as a VSS dealer for a polynomial
\(f_j(X)\).  The servers verify the received evaluations against
homomorphic commitments to the coefficients.  Once invalid dealers, if
any, have been excluded, server \(S_i\) adds all accepted evaluations and
obtains
\[
  s_i^{(0)}
  =\sum_{j\in\mathcal Q}f_j(i)
  =F_0(i),
  \qquad
  F_0(X)=\sum_{j\in\mathcal Q}f_j(X).
\]
The resulting master key is the constant term,
$$k=F_0(0)=\sum_{j\in\mathcal Q}a_{j,0}$$  No server needs to assemble
this sum explicitly.

The coefficient commitments certify the Shamir sharing, but they do not
by themselves show that the public isogeny key was formed from the same
constant coefficients.  The link-proof chain closes this gap.

Each
qualified dealer applies its constant-term action to the current curve
and proves, with one common witness, that the action exponent is the
value committed in \(A_{j,0}\).  The final curve is therefore
\(\pk=[k]E_0\).  The epoch-zero certificate records the public key,
aggregate coefficient commitments, qualified set, and supporting
transcript as the initial public state of the service.

\begin{figure}[H]
\centering
\fbox{\begin{minipage}{0.92\textwidth}
\small
\textbf{Protocol 1: Dealerless Distributed Key Generation}\\[4pt]
\noindent\textbf{Parties}: Servers $S_1,\dots,S_n$.\quad
\textbf{Output}: Shamir shares $\{s_i^{(0)}\}$ of master key $k$,
public key $\pk=[k]E_0$, epoch certificate $\cert_0$.\\[2pt]
\rule{\textwidth}{0.4pt}\\[2pt]
\begin{enumerate}[nosep,leftmargin=*,label=\textbf{D\arabic*.}]
  \item \textbf{Polynomial generation.}  Each $S_j$ samples
        $a_{j,0},\dots,a_{j,t-1}\leftarrow\Zq$ and sets $f_j(X)=
        \sum_{\ell=0}^{t-1}a_{j,\ell}X^\ell$.

  \item \textbf{Coefficient commitment broadcast.}  $S_j$ samples
        $\rho_{j,\ell}\leftarrow\R$ and broadcasts
        $A_{j,\ell}=\Com(a_{j,\ell};\rho_{j,\ell})$ for all $\ell$.

  \item \textbf{Private share distribution.}  For each receiver $S_i$,
        $S_j$ privately sends
        $u_{j,i}=f_j(i)$ and $\varrho_{j,i}=
        \sum_{\ell=0}^{t-1}i^\ell\rho_{j,\ell}$.

  \item \textbf{Share verification.}  $S_i$ accepts if and only if
        $\Com(u_{j,i};\varrho_{j,i})=
        \bigoplus_{\ell=0}^{t-1}i^\ell\odot A_{j,\ell}$.
        Complaints are resolved via broadcast, and disqualified dealers are
        removed from the qualified set $\mathcal{Q}$.

  \item \textbf{Share aggregation.}  Each $S_i$ computes
        \begin{align*}
          s_i^{(0)} &= \sum_{j\in\mathcal{Q}} u_{j,i} \pmod q,\qquad
          \omega_i^{(0)} = \sum_{j\in\mathcal{Q}} \varrho_{j,i}.
        \end{align*}
        Public aggregate coefficients:
        $A_\ell^{(0)}=\bigoplus_{j\in\mathcal{Q}}A_{j,\ell}$.
        The implicit master key is $k=\sum_{j\in\mathcal{Q}}a_{j,0}$.

  \item \textbf{Link proofs.}  Order $\mathcal{Q}=(j_1,\dots,j_m)$
        arbitrarily.  Set $P_0=E_0$.  For $h=1,\dots,m$, server~$S_{j_h}$
        computes $P_h=[a_{j_h,0}]P_{h-1}$,
        generates $\pi_h^{\mathsf{link}}\leftarrow
        \NIZK.\Prove_{\mathcal{R}_{\mathsf{link}}}
        (A_{j_h,0},P_{h-1},P_h;\, a_{j_h,0},\rho_{j_h,0})$,
        and broadcasts $(P_h,\pi_h^{\mathsf{link}})$.
        A dealer whose proof fails is removed.  The public key is
        $\pk=P_m=[k]E_0$.

  \item \textbf{Epoch-zero certificate.}  Servers sign
        $\cert_0 = (e,\ctx,\pk,\mathbf A^{(0)},\mathcal{Q},
                     \{\pi_h^{\mathsf{link}}\}_{h=1}^{|\mathcal{Q}|},
                     \mathsf{digest}_{\mathsf{VSS}})$.
\end{enumerate}
\end{minipage}}
\caption{Dealerless distributed key generation for \scheme.}
\label{fig:prot-dkg}
\end{figure}

\subsection{Protocol 2: Threshold evaluation}
\label{sec:prot2}

Once the initial state has been established, clients can evaluate the
keyed function without revealing their inputs and without causing the
servers to reconstruct the master key.  The central idea is to combine
input blinding with Lagrange interpolation in the exponent of the group
action.

For an input \(x\), the client first maps the domain-separated input to a
curve
\[
  X=\Hone(\ctx\|x).
\]
It then chooses a fresh scalar \(r\) and sends the blinded curve
\(B=[r]X\).  Because the action is free and transitive, a uniformly
sampled blinding action hides the original curve from the servers.

The
blinding proof is included to show that the submitted curve is a
well-formed blinding of a hash-derived input, rather than an arbitrary
curve selected to turn the committee into an unrestricted action oracle.

The selected quorum \(I=\{i_1,\ldots,i_t\}\) does not reconstruct \(k\)
as a scalar.  Instead, server \(S_i\) applies the action associated with
its Lagrange-weighted share
\(\lambda_i^I s_i^{(e)}\).  Since group actions compose additively, the
sequential chain accumulates the sum
\[
  \sum_{i\in I}\lambda_i^I s_i^{(e)}=k.
\]
The final server therefore returns \([k]B\).  The client applies
\([-r]\) and obtains \([k]X\), from which the final OPRF output is derived.

The order of the servers is operational rather than algebraically
significant because the underlying group is abelian.  Nevertheless, a
fixed order is included in the session transcript so that every partial
action has an unambiguous predecessor and successor.  Each server signs
its transition and proves that it used the share certified for the
current epoch.  As a result, the client can verify not only the final
curve, but the complete sequence by which that curve was produced.
\begin{figure}[H]
\centering
\fbox{\begin{minipage}{0.92\textwidth}
\small
\textbf{Protocol 2: Threshold Evaluation}\\[4pt]
\noindent\textbf{Parties}: Client $C$ with private input $x$,
quorum $I=\{i_1,\dots,i_t\}$ in epoch~$e$.\quad
\textbf{Output}: Client learns $y = F_k(x)$.\\[2pt]
\rule{\textwidth}{0.4pt}\\[2pt]
\begin{enumerate}[nosep,leftmargin=*,label=\textbf{E\arabic*.}]
  \item \textbf{Client blinding.}  $C$ verifies $\cert_e$, computes
        $X=\Hone(\ctx\|x)$, samples $r\leftarrow\Zq$, sets $B=[r]X$,
        and generates $\pi^{\mathsf{blind}}\leftarrow
        \NIZK.\Prove_{\mathcal{R}_{\mathsf{blind}}}
        (\ctx,\sid,B;\,x,r)$.  The client sends
        $\mathsf{req}=(\sid,e,I,H(\cert_e),B,\pi^{\mathsf{blind}})$ to $S_{i_1}$.
        The quorum $I\subseteq[n]$ of size $|I|=t$ is selected via the
        deterministic, publicly verifiable function
        $\mathsf{SelectQuorum}(\sid,e,n,t)$, instantiated as a
        pseudorandom permutation over $[n]$ seeded by a
        collision-resistant hash $H(\sid\|e\|\ctx)$.  Every party can
        independently recompute $I$ and verify its correctness,
        preventing a malicious client or coordinator from biasing server
        selection across retries.

  \item \textbf{Request validation.}  $S_{i_1}$ verifies $\pi^{\mathsf{blind}}$,
        checks that $\sid$ is fresh, $e$ is the current epoch, $H(\cert_e)$ matches its active epoch certificate, and $B\in\E$.
        It rejects on failure and sets $Q_0\leftarrow B$.

  \item \textbf{Sequential partial evaluation.}  For $h=1,\dots,t$,
        server~$S_{i_h}$,
        \begin{enumerate}[nosep]
          \item Computes $\lambda_{i_h}^I=\prod_{j\in I,\,j\neq i_h}
                (-j)/(i_h-j)\pmod q$.
          \item Computes $Q_h=[\lambda_{i_h}^I\cdot
                s_{i_h}^{(e)}]\,Q_{h-1}$.
          \item Generates $\pi_h^{\mathsf{eval}}\leftarrow
                \NIZK.\Prove_{\mathcal{R}_{\mathsf{eval}}}
                (\mathbf A^{(e)},i_h,Q_{h-1},Q_h,\lambda_{i_h}^I;\,
                s_{i_h}^{(e)},\omega_{i_h}^{(e)})$.
          \item Signs $M_h=(\sid,e,I,h,Q_{h-1},Q_h,\pi_h^{\mathsf{eval}})$
                as $\sigma_h$.
          \item Forwards $(\sid,e,I,Q_0,\{Q_\ell,\pi_\ell^{\mathsf{eval}},\sigma_\ell\}_{\ell=1}^{h})$
                to $S_{i_{h+1}}$ (or to $C$ if $h=t$).
          \item Before acting, $S_{i_{h+1}}$ checks the epoch and certificate hash and verifies all prior signatures
                and proofs.
        \end{enumerate}

  \item \textbf{Client unblinding.}  $C$ receives
        $\mathsf{resp}=(Q_t,\{Q_h,\pi_h^{\mathsf{eval}},\sigma_h\}_{h=1}^{t},
        \cert_e)$.  The client recomputes all $\lambda_i^I$, verifies every
        $\pi_h^{\mathsf{eval}}$ and $\sigma_h$, recomputes each
        share commitment $\EvalCom(\mathbf A^{(e)},i_h)$ from the
        coefficient vector $\mathbf A^{(e)}$ contained in $\cert_e$,
        and computes,
        \begin{align*}
          Y &= [-r]Q_t,\\
          y &= \Htwo(\mathsf{PIVOT-out}\|\ctx\|\pk\|x\|
                     \mathsf{enc}(Y)).
        \end{align*}
        The epoch number is excluded from the output hash so that OPRF
        values remain invariant across refresh boundaries.
\end{enumerate}
\end{minipage}}
\caption{Threshold evaluation protocol for \scheme.}
\label{fig:prot-eval}
\end{figure}
\begin{figure}[H]
\centering
\begin{tikzpicture}[
  node distance=1.8cm,
  server/.style={rectangle, draw, rounded corners, minimum height=1cm, minimum width=1.2cm, fill=blue!8},
  client/.style={rectangle, draw, rounded corners, minimum height=1cm, minimum width=1.5cm, fill=green!8},
  curve/.style={circle, draw, minimum size=0.7cm, fill=yellow!10, font=\small},
  arr/.style={-{Stealth[length=2.5mm]}, thick},
  lbl/.style={font=\footnotesize, midway}
]
  \node[client] (C) {Client $C$};
  
  \node[curve, right=1.2cm of C] (Q0) {$Q_0$};
  \node[server, above=0.4cm of Q0] (S1) {$S_{i_1}$};
  \node[curve, right=1.5cm of Q0] (Q1) {$Q_1$};
  \node[server, above=0.4cm of Q1] (S2) {$S_{i_2}$};
  \node[right=1cm of Q1] (dots) {$\cdots$};
  \node[curve, right=1cm of dots] (Qt1) {$Q_{t{-}1}$};
  \node[server, above=0.4cm of Qt1] (St) {$S_{i_t}$};
  \node[curve, right=1.5cm of Qt1] (Qt) {$Q_t$};
  
  \draw[arr] (C) -- node[lbl, above] {$B{=}[r]X$} (Q0);
  \draw[arr] (Q0) -- node[lbl, above] {$[\lambda_1 s_1]$} (Q1);
  \draw[arr] (Q1) -- (dots);
  \draw[arr] (dots) -- (Qt1);
  \draw[arr] (Qt1) -- node[lbl, above] {$[\lambda_t s_t]$} (Qt);
  \draw[arr, dashed] (Qt) -- ++(0,-1.2) -| node[lbl, below, pos=0.25] {$Y{=}[-r]Q_t{=}[k]X$} (C);
  
  \draw[arr, dotted] (S1) -- (Q0);
  \draw[arr, dotted] (S2) -- (Q1);
  \draw[arr, dotted] (St) -- (Qt1);
\end{tikzpicture}
\caption{Sequential evaluation chain in \scheme.  The client sends a blinded
input $Q_0=B=[r]X$ to the first server.  Each server $S_{i_h}$ applies its
Lagrange-weighted share to produce $Q_h=[\lambda_{i_h}^I s_{i_h}]Q_{h-1}$,
accompanied by a NIZK proof $\pi_h^{\mathsf{eval}}$.  The accumulated
result $Q_t=[k]B$ is returned to the client, who unblinds to recover
$Y=[k]X$.}
\label{fig:eval-chain}
\end{figure}
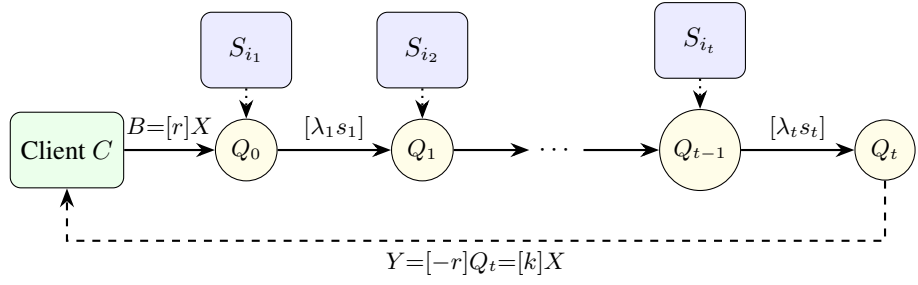

\begin{lemma}[Evaluation correctness]
For any epoch $e$, if all servers in quorum $I$ honestly apply their
certified shares, the client's output equals $F_k(x)$ as defined in
Equation~\ref{eq:prf}.
\end{lemma}
\begin{proof}
By induction on $h$, $$Q_h=\left[\sum_{j=1}^{h}\lambda_{i_j}^I
s_{i_j}^{(e)}\right]B$$ using $[a]([b]E)=[a+b]E$ from Section~\ref{sec:group-action}.
For $h=t$, $$Q_t=\left[\sum_{i\in I}\lambda_i^I s_i^{(e)}\right]B=[k]B$$
Since $B=[r]X$, the composition law $[a]([b]E)=[a+b]E$ gives
$Q_t=[k]([r]X)=[k+r]X$. Then $Y=[-r]Q_t=[-r]([k+r]X)=[{-r}+k+r]X=[k]X$ by the same law.
Hence $$y=\Htwo(\mathsf{PIVOT-out}\|\ctx\|\pk\|x\|\mathsf{enc}([k]X))
=F_k(x)$$
\end{proof}

\subsection{Protocol 3: Proactive share refresh}
\label{sec:prot3}

Threshold sharing protects the master key only as long as the adversary
does not collect \(t\) shares belonging to the same sharing polynomial.
In a long-lived service, this requirement cannot be enforced merely by
assuming that one fixed set of servers remains honest forever.  A mobile
adversary may compromise different servers at different times and retain
the shares that it learns.  Proactive refresh addresses this threat by
periodically replacing the current shares with fresh shares of the same
secret.

The refresh mechanism relies on a simple algebraic observation.  Adding a
polynomial \(z(X)\) satisfying \(z(0)=0\) to the current sharing
polynomial changes every nonconstant coefficient but leaves the secret
unchanged,
\[
  F_{e+1}(X)=F_e(X)+z(X),
  \qquad
  F_{e+1}(0)=F_e(0)=k.
\]
To avoid placing trust in a single refresh dealer, every server
contributes its own random zero polynomial.  The qualified contributions
are added together, so that one honest contribution is sufficient to
re-randomize the nonconstant coefficients from the adversary's
perspective. The coefficient commitments evolve in parallel with the shares.
There is deliberately no commitment to a refresh constant term, zero is enforced syntactically by defining the refresh polynomial only
from powers \(X,\ldots,X^{t-1}\).  Consequently, the commitment
\(A_0^{(e)}\) remains unchanged, whereas the higher-degree commitments
are updated homomorphically.

Refresh also requires a careful transition rule.  It is not sufficient
for servers to compute correct next-epoch shares independently, they
must agree on the same commitment vector and must not combine shares
from different epochs in one evaluation. The PROPOSE, PREPARE, COMMIT,
and ACTIVATE stages establish a common certificate for the next state.

During the transition, the old state remains the active state.  A server
begins using \(s_i^{(e+1)}\) only after accepting the common next-epoch
certificate.  Evaluation requests are bound to an epoch number and
certificate, so a mixed-epoch chain is rejected rather than interpolated. Finally, after activation, honest servers erase the old share, old
opening randomness, received zero-shares, and temporary refresh state.
The algebraic refresh creates independence between successive sharing
polynomials and secure erasure ensures that a later corruption cannot
recover the obsolete local state. At the boundary from epoch \(e\) to epoch \(e+1\), the servers execute
the following procedure,

\begin{figure}[H]
\centering
\fbox{\begin{minipage}{0.92\textwidth}
\small
\textbf{Protocol 3: Proactive Share Refresh}\\[4pt]
\noindent\textbf{Parties}: All $n$ servers at epoch boundary $e \to e{+}1$.\quad
\textbf{Output}: Fresh shares $\{s_i^{(e+1)}\}$ of the same master key~$k$.\\[2pt]
\rule{\textwidth}{0.4pt}\\[2pt]
\begin{enumerate}[nosep,leftmargin=*,label=\textbf{R\arabic*.}]
  \item \textbf{Zero-polynomial generation.}  Each refresh dealer $S_j$
        samples $b_{j,1},\dots,b_{j,t-1}\leftarrow\Zq$ and constructs
        $z_j(X)=\sum_{\ell=1}^{t-1}b_{j,\ell}X^\ell$.
        The constant term is syntactically zero, no $D_{j,0}$ is published.
        $S_j$ broadcasts $D_{j,\ell}=\Com(b_{j,\ell};\eta_{j,\ell})$
        for $\ell=1,\dots,t-1$.

  \item \textbf{Zero-share distribution.}  $S_j$ sends to each $S_i$:
        $\delta_{j,i}=z_j(i)$ and $\nu_{j,i}=
        \sum_{\ell=1}^{t-1}i^\ell\eta_{j,\ell}$.

  \item \textbf{Zero-share verification.}  $S_i$ accepts if and only if
        \begin{equation}\label{eq:zero-verify}
          \Com(\delta_{j,i};\nu_{j,i}) =
          \bigoplus_{\ell=1}^{t-1} i^\ell\odot D_{j,\ell}.
        \end{equation}
        Complaints are resolved as in the DKG.  Let $\mathcal{R}_e$ denote the
        qualified refresh-dealer set.

  \item \textbf{Share and commitment update.}
        \begin{align*}
          s_i^{(e+1)} &= s_i^{(e)} + \sum_{j\in\mathcal{R}_e}\delta_{j,i}
                        \pmod q,\qquad
          \omega_i^{(e+1)} = \omega_i^{(e)} + \sum_{j\in\mathcal{R}_e}
                             \nu_{j,i},\\
          A_0^{(e+1)} &= A_0^{(e)}\quad\text{(unchanged)},\qquad
          A_\ell^{(e+1)} = A_\ell^{(e)} \oplus
                           \bigoplus_{j\in\mathcal{R}_e} D_{j,\ell},
                           \quad 1\le\ell<t.
        \end{align*}

  \item \textbf{Coordinated epoch transition.}
        A designated epoch leader $S_{\ell}$ (e.g., $S_{(e\bmod n)+1}$,
        rotating per epoch) coordinates the transition,
        \begin{description}[nosep,leftmargin=1em]
          \item[PROPOSE.] $S_{\ell}$ computes the proposed
                certificate $\cert'_{e+1}=(\pk,\mathbf A^{(e+1)},e{+}1,
                \mathcal{R}_e,\mathsf{digest}_{\mathsf{complaints}})$
                and broadcasts it with a signature.  Honest servers
                independently verify $\cert'_{e+1}$ against their local
                updates.
          \item[PREPARE.] Each server that verifies
                $\cert'_{e+1}$ broadcasts a signed
                $\mathsf{PREPARE}(e,e{+}1,H(\cert'_{e+1}))$.
                During the PREPARE state, servers continue responding to
                evaluation requests with epoch-$e$ shares.
                If a server detects inconsistency, it broadcasts a
                complaint and the leader is replaced.
          \item[COMMIT.] Upon receiving $t$ valid PREPARE
                messages for the same certificate hash,
                each server broadcasts a signed
                $\mathsf{COMMIT}(e,e{+}1,h)$ and enters COMMIT-pending
                state, still responding with epoch-$e$ shares.
          \item[ACTIVATE.] Upon receiving $t$ valid COMMIT messages,
                each server advances $\mathsf{epoch}\leftarrow e{+}1$,
                signs $\cert_{e+1}$, and begins responding with
                epoch-$(e{+}1)$ shares.  No server uses epoch-$(e{+}1)$
                shares before this point, preventing mixed-epoch quorums.
        \end{description}

  \item \textbf{Secure erasure.}  Honest servers securely erase
        $s_i^{(e)}$, $\omega_i^{(e)}$, all received $\delta_{j,i}$,
        and all refresh randomness.
\end{enumerate}
\end{minipage}}
\caption{Proactive share refresh protocol for \scheme.}
\label{fig:prot-refresh}
\end{figure}

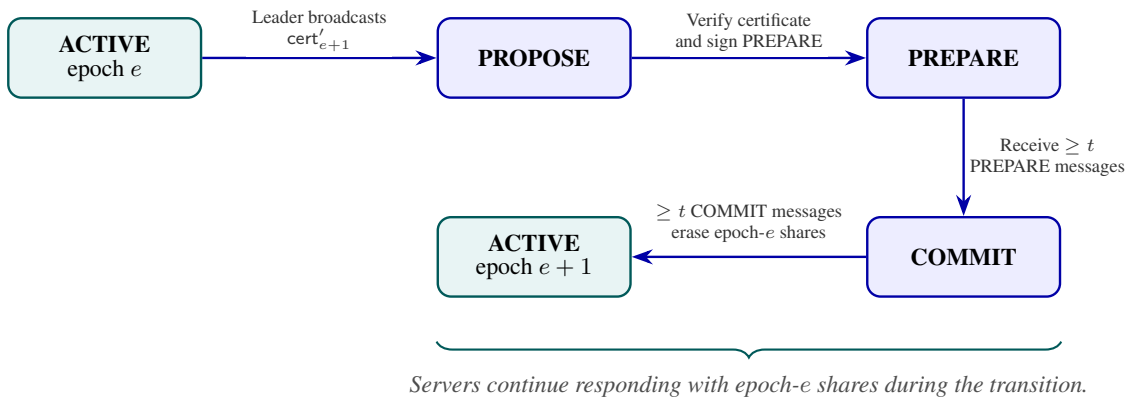
\begin{figure}[H]
  \centering
  \begin{tikzpicture}[
    node distance=1.55cm and 3.10cm,
    state/.style={
      rectangle,
      rounded corners=5pt,
      draw=blue!65!black,
      line width=0.9pt,
      fill=blue!7,
      minimum width=2.55cm,
      minimum height=1.05cm,
      align=center,
      font=\small\bfseries
    },
    active state/.style={
      state,
      draw=teal!70!black,
      fill=teal!9
    },
    arrow/.style={
      -{Stealth[length=2.8mm,width=2mm]},
      line width=0.9pt,
      draw=blue!65!black
    },
    label/.style={
      font=\scriptsize,
      align=center,
      text=black!80,
      fill=white,
      inner sep=2pt
    },
    note/.style={
      font=\footnotesize\itshape,
      align=center,
      text=black!70
    }
  ]
    \node[active state] (active) {ACTIVE\\[-1pt]
      {\footnotesize\mdseries epoch $e$}};
    \node[state, right=of active] (propose) {PROPOSE};
    \node[state, right=of propose] (prepare) {PREPARE};
    \node[state, below=of prepare] (commit) {COMMIT};
    \node[active state, left=of commit] (next) {ACTIVE\\[-1pt]
      {\footnotesize\mdseries epoch $e+1$}};

    \draw[arrow] (active) -- node[label, above] {
      Leader broadcasts\\$\cert'_{e+1}$} (propose);

    \draw[arrow] (propose) -- node[label, above] {
      Verify certificate\\and sign PREPARE} (prepare);

    \draw[arrow] (prepare) -- node[label, right] {
      Receive $\geq t$\\PREPARE messages} (commit);

    \draw[arrow] (commit) -- node[label, above=3pt] {
      $\geq t$ COMMIT messages\\erase epoch-$e$ shares} (next);

    \draw[
      decorate,
      decoration={brace, amplitude=6pt, mirror},
      draw=teal!70!black,
      line width=0.8pt
    ]
      ([yshift=-0.58cm]next.south west) --
      ([yshift=-0.58cm]commit.south east)
      node[note, midway, below=9pt] {
        Servers continue responding with epoch-$e$ shares during the transition.};
  \end{tikzpicture}
\caption{Epoch transition state machine for proactive refresh.  Servers
continue responding with epoch-$e$ shares during the PROPOSE, PREPARE,
and COMMIT phases, switching to epoch-$(e{+}1)$ shares only after
receiving $t$ valid COMMIT messages.}
\label{fig:epoch-fsm}
\end{figure}

\begin{lemma}[Refresh correctness]
After Protocol~3, $F_{e+1}(0)=k$ and $A_0^{(e+1)}=A_0^{(e)}$.
\end{lemma}
\begin{proof}
$F_{e+1}(X)=F_e(X)+\sum_{j\in\mathcal{R}_e}z_j(X)$.  We must show that
$z_j(0)=0$ for every qualified dealer $j\in\mathcal{R}_e$.

The verification check (\Cref{eq:zero-verify}) requires
$\Com(\delta_{j,i};\nu_{j,i})=\bigoplus_{\ell=1}^{t-1}i^\ell\odot
D_{j,\ell}$.  By the homomorphic property, the right-hand side equals
$\Com(\sum_{\ell=1}^{t-1}b_{j,\ell}i^\ell;\;\cdot\;)$.  The binding
property of $\Com$  then forces
$\delta_{j,i}=\sum_{\ell=1}^{t-1}b_{j,\ell}i^\ell$ for every receiver
$S_i$ that accepts.  Since every accepted share lies on the polynomial
$\hat z_j(X)=\sum_{\ell=1}^{t-1}b_{j,\ell}X^\ell$ and at least $t$
receivers accept (otherwise $S_j$ is disqualified), the unique polynomial
of degree at most $t-1$ through these $t$ points is $\hat z_j$ itself.
In particular, $z_j(0)=\hat z_j(0)=0$.

Hence $F_{e+1}(0)=F_e(0)+\sum_{j}z_j(0)=k+0=k$, and
$A_0^{(e+1)}=A_0^{(e)}$ since the zero-polynomial constant terms
contribute nothing.
\end{proof}

\subsection{Protocol 4: Verification, blame, and robust restart}
\label{sec:prot4}

The preceding protocols describe the computations performed during an
honest execution.  In a distributed deployment, however, an evaluation
or refresh may fail because a party submits a malformed proof, applies an
incorrect action, distributes an inconsistent share, or simply refuses
to continue.  Verifiability is useful only when such failures can be
attributed to a specific signed message and when honest parties have a
well-defined recovery procedure.

\scheme{} therefore separates \emph{detection}, \emph{attribution}, and
\emph{recovery}.  Detection is local, the next server in an evaluation
chain, the client, or a refresh recipient checks the relevant NIZK,
signature, or commitment equation.  Attribution is public, the detecting
party packages the signed offending message together with the failed
verification instance.  Any observer can then check the complaint
without learning an honest share.  Recovery is handled by removing the
identified server from the eligible pool and restarting the session with
a fresh identifier and a newly selected quorum.

A restart uses a fresh \(\sid\) because signatures, proofs, and quorum
selection are session-bound.  Reusing the failed identifier would make
the new execution difficult to distinguish from a continuation or replay
of the old one.  The wrapper provides progress as long as at least \(t\)
eligible servers remain.  It does not guarantee availability once the
active pool falls below the threshold, in that case, abort is the only
correct outcome.

\begin{figure}[H]
\centering
\fbox{\begin{minipage}{0.92\textwidth}
\small
\textbf{Protocol 4: Verification, Blame, and Robust Restart}\\[4pt]
\noindent\textbf{Purpose}: Identify and remove misbehaving servers and
ensure liveness through retry.\\[2pt]
\rule{\textwidth}{0.4pt}\\[2pt]
\begin{enumerate}[nosep,leftmargin=*,label=\textbf{V\arabic*.}]
  \item \textbf{Invalid blinding proof.}  A server detecting an invalid
        $\pi^{\mathsf{blind}}$ broadcasts a signed rejection
        $(\sid,\pi^{\mathsf{blind}},\text{``invalid blinding''})$.
        No secret material is involved.

  \item \textbf{Invalid server evaluation.}  If $\pi_h^{\mathsf{eval}}$
        for $S_{i_h}$ fails verification (detected by $S_{i_{h+1}}$ or
        the client), the detecting party produces a blame certificate
        $\mathsf{blame}=(\sid,e,I,h,Q_{h-1},Q_h,
        \pi_h^{\mathsf{eval}},\sigma_h)$.  Any third party can verify
        that $\sigma_h$ is valid yet $\pi_h^{\mathsf{eval}}$ is invalid.
        The blamed server is removed from the active pool.

  \item \textbf{Invalid refresh contribution.}  A complaint against a
        refresh dealer includes the signed
        $(\delta_{j,i},\nu_{j,i})$ and the public $\{D_{j,\ell}\}$, any
        third party verifies \Cref{eq:zero-verify}.

  \item \textbf{Robust wrapper.}  Maintain a pool of $n$ certified
        servers.  On evaluation failure, remove the blamed server,
        generate a fresh $\sid$, and retry with a new quorum via
        $\mathsf{SelectQuorum}(\sid,e,n,t)$.  If fewer than $t$ servers
        remain, abort.  Otherwise, the wrapper completes within at most
        $n-t+1$ retries.  The wrapper assumes that the total number of
        servers satisfies $n \ge 2t$, so that removing up to $t-1$
        misbehaving servers still leaves at least $t$ honest servers
        available for a successful evaluation.
\end{enumerate}
\end{minipage}}
\caption{Verification, blame, and robust restart for \scheme.}
\label{fig:prot-blame}
\end{figure}
\subsection{Protocol 5: Committee resharing}
\label{sec:prot5}
Proactive refresh protects a fixed committee over time, but it does not
address changes in membership.  Long-lived services may need to replace
failed machines, rotate administrative domains, increase the committee
size, or adopt a new threshold.  Re-running the DKG would generate a new
master key and would therefore change the public key and all OPRF
outputs.  Committee resharing instead transfers the existing secret
\(k\) to a new Shamir sharing without ever reconstructing \(k\).

Let \(I\) be a qualified old-committee quorum.  The Lagrange-weighted
values \(\{\lambda_i^I s_i^{(e)}\}_{i\in I}\) satisfy
\[
  \sum_{i\in I}\lambda_i^I s_i^{(e)}=k.
\]
Each old server \(S_i\) treats its weighted value as the constant term of
a fresh degree-at-most-\((t'-1)\) polynomial \(g_i(X)\) and distributes
that polynomial to the new committee using VSS.  The new servers add the
received evaluations.  Their aggregate polynomial
\[
  G(X)=\sum_{i\in I}g_i(X)
\]
has constant term \(G(0)=k\), while its nonconstant coefficients are
freshly randomized for the new threshold \(t'\).

The reshare proof is necessary because a commitment to the constant term
of \(g_i\) must be linked to the old server's certified share.  Without
this link, an old server could distribute a perfectly consistent
polynomial having an unrelated constant term, thereby changing the
aggregate secret.  Once the proofs and VSS checks have been accepted,
the new committee signs a certificate for its aggregate commitment
vector.  The public key remains \(\pk=[k]E_0\), and previously computed
OPRF outputs remain valid. Protocol~5 therefore migrates the sharing of \(k\) from the old committee
to a new committee \((S'_1,\dots,S'_{n'})\) with threshold \(t'\), without
revealing or reconstructing the master key.

\begin{figure}[H]
\centering
\fbox{\begin{minipage}{0.92\textwidth}
\small
\textbf{Protocol 5: Committee Resharing}\\[4pt]
\noindent\textbf{Parties}: Old committee $(S_1,\dots,S_n)$ with threshold~$t$,
new committee $(S'_1,\dots,S'_{n'})$ with threshold~$t'$.\quad
\textbf{Output}: New Shamir shares $\{s_j'\}$ of the same master key~$k$.\\[2pt]
\rule{\textwidth}{0.4pt}\\[2pt]
\begin{enumerate}[nosep,leftmargin=*,label=\textbf{M\arabic*.}]
  \item \textbf{Quorum selection.} Select a quorum $I$ of size $\ge t$ from
        the old committee.

  \item \textbf{Reshare VSS.} Each old-committee server $S_i$ ($i\in I$)
        acts as a VSS dealer and 
        samples $g_i(X)=\lambda_i^I s_i^{(e)}+\sum_{\ell=1}^{t'-1}
        c_{i,\ell}X^\ell$ (constant term $\lambda_i^I s_i^{(e)}$), publishes
        coefficient commitments $B_{i,\ell}$ with sampled randomness, and
        generates $\pi_i^{\mathsf{reshare}}\leftarrow
        \NIZK.\Prove_{\mathcal{R}_{\mathsf{reshare}}}
        (\mathbf A^{(e)},i,B_{i,0},\lambda_i^I;\,
        s_i^{(e)},\omega_i^{(e)},\rho_{i,0})$.

  \item \textbf{New share aggregation.} Each new server $S_j'$ receives
        $g_i(j)$ and its commitment randomness from every $i\in I$, verifies against $\{B_{i,\ell}\}$,
        and sets $s_j'=\sum_{i\in I} g_i(j)$.

  \item \textbf{Correctness.} The aggregate polynomial
        $G(X)=\sum_{i\in I}g_i(X)$ satisfies
        $G(0)=\sum_{i\in I}\lambda_i^I s_i^{(e)}=k$.  The new committee
        signs $\cert_{e+1}$ after confirmation.
\end{enumerate}
\end{minipage}}
\caption{Committee resharing protocol for \scheme.}
\label{fig:prot-reshare}
\end{figure}

\section{Correctness Analysis}
\label{sec:correctness}

This section establishes the functional correctness of \scheme{} under
honest execution.  The purpose of the analysis is to show that the
distributed key-generation, evaluation, proactive-refresh, and
committee-resharing procedures preserve a single well-defined master key
throughout the lifetime of the system.  In particular, every successful
evaluation must produce the same value as a direct evaluation of the
underlying function \(F_k\) with the master key \(k\).

Correctness is distinct from security.  The arguments below assume that
all participating parties follow the prescribed algorithms and that all
messages required by a successful execution are delivered without
modification.  Privacy against corrupted parties, simulation of protocol
transcripts, robustness against malformed contributions, and
identifiable aborts are addressed separately in the security analysis.

We use the following standard assumptions throughout this section.

\begin{enumerate}[nosep,label=(\roman*)]
  \item The server identifiers \(1,\ldots,n\) are distinct nonzero
        elements of \(\Zq\), and \(n<q\).  Hence every denominator in the
        Lagrange coefficients is invertible in \(\Zq\).
  \item The group action satisfies
        \([0]E=E\) and \([a]([b]E)=[a+b]E\) for all
        \(a,b\in\Zq\) and \(E\in\E\).
  \item The commitment scheme is correct and additively homomorphic.
  \item The NIZK proof system is complete, and the signature scheme is
        correct.
  \item Every evaluation session is executed using shares and
        coefficient commitments belonging to one common epoch.
\end{enumerate}

The final condition is essential.  Shares taken from two different
epochs generally lie on two different Shamir polynomials, even though
both polynomials have the same constant term.  Consequently, Lagrange
interpolation over a mixed-epoch quorum need not recover \(k\).  The
protocol must therefore require every server in an evaluation chain to
check the epoch number and the hash of the corresponding epoch
certificate before applying its share.

The correctness of the complete construction is most naturally expressed
through an invariant that is maintained by the DKG, refresh, and
resharing procedures.

\begin{definition}[Valid epoch state]
\label{def:valid-epoch-state}
An epoch-\(e\) state is said to be valid for threshold \(t\) and master
key \(k\) if there exist coefficients
\(a_0^{(e)},\ldots,a_{t-1}^{(e)}\in\Zq\) and commitment randomness
\(\rho_0^{(e)},\ldots,\rho_{t-1}^{(e)}\) such that, for
\[
  F_e(X)=\sum_{\ell=0}^{t-1}a_\ell^{(e)}X^\ell,
\]
the following conditions hold,
\begin{enumerate}[nosep,label=(\roman*)]
  \item \(F_e(0)=a_0^{(e)}=k\),
  \item Every server \(S_i\) holds
        \(s_i^{(e)}=F_e(i)\),
  \item The public coefficient commitments satisfy
        \[
          A_\ell^{(e)}
          =\Com(a_\ell^{(e)};\rho_\ell^{(e)})
          \qquad\text{for }0\leq \ell<t
        \]
  \item For every server index \(i\),
        \[
          \EvalCom(\mathbf A^{(e)},i)
          =\Com(s_i^{(e)};\omega_i^{(e)}),
          \qquad
          \omega_i^{(e)}
          =\sum_{\ell=0}^{t-1}i^\ell\rho_\ell^{(e)}
        \]
  \item The public key is \(\pk=[k]E_0\).
\end{enumerate}
\end{definition}

The polynomial \(F_e\) may have degree strictly smaller than \(t-1\) if
its highest coefficients cancel.  Thus, throughout this section,
``degree-\((t-1)\) Shamir sharing'' means a sharing defined by a
polynomial of degree at most \(t-1\).

\subsection{Correctness of dealerless distributed key generation}

Protocol~1 must produce an initial valid epoch state without any party
explicitly reconstructing the master key.

\begin{lemma}[Correctness of the DKG]
\label{lem:dkg-correct}
Assume that all servers execute Protocol~1 honestly and let
\(\mathcal Q\) denote the resulting qualified dealer set.  Define
\[
  f_j(X)=\sum_{\ell=0}^{t-1}a_{j,\ell}X^\ell
  \qquad\text{for each }j\in\mathcal Q.
\]
Then Protocol~1 produces a valid epoch-\(0\) state with
\[
  F_0(X)=\sum_{j\in\mathcal Q}f_j(X)
  \quad\text{and}\quad
  k=F_0(0)=\sum_{j\in\mathcal Q}a_{j,0}.
\]
Moreover, if \(\mathcal Q\neq\varnothing\), then \(k\) is uniformly
distributed in \(\Zq\).
\end{lemma}

\begin{proof}
For each qualified dealer \(S_j\), the polynomial
\(f_j(X)\) has degree at most \(t-1\).  Dealer \(S_j\) sends
\[
  u_{j,i}=f_j(i)
  \quad\text{and}\quad
  \varrho_{j,i}
  =\sum_{\ell=0}^{t-1}i^\ell\rho_{j,\ell}
\]
to server \(S_i\).  Because all parties are honest, the share-verification
equation holds,
\begin{align*}
  \Com(u_{j,i};\varrho_{j,i})
  &=
  \bigoplus_{\ell=0}^{t-1}
      i^\ell\odot A_{j,\ell}                                      \\
  &=
  \bigoplus_{\ell=0}^{t-1}
      i^\ell\odot\Com(a_{j,\ell};\rho_{j,\ell})                    \\
  &=
  \Com\!\left(
      \sum_{\ell=0}^{t-1}a_{j,\ell}i^\ell;
      \sum_{\ell=0}^{t-1}\rho_{j,\ell}i^\ell
      \right)                                                      \\
  &=
  \Com(f_j(i);\varrho_{j,i}).
\end{align*}
Thus every accepted contribution is the evaluation of the polynomial
committed to by the dealer.

Each server aggregates its received values as
\[
  s_i^{(0)}
  =\sum_{j\in\mathcal Q}u_{j,i}
  =\sum_{j\in\mathcal Q}f_j(i)
  =F_0(i).
\]
Since a sum of polynomials of degree at most \(t-1\) again has degree at
most \(t-1\), the values
\(\{s_i^{(0)}\}_{i=1}^{n}\) form a valid threshold-\(t\) Shamir sharing
of
\[
  F_0(0)=\sum_{j\in\mathcal Q}f_j(0)
        =\sum_{j\in\mathcal Q}a_{j,0}
        =k.
\]

We next verify consistency of the aggregate coefficient commitments.
For every \(\ell\),
\[
  A_\ell^{(0)}
  =\bigoplus_{j\in\mathcal Q}A_{j,\ell}
  =\Com\!\left(
      \sum_{j\in\mathcal Q}a_{j,\ell};
      \sum_{j\in\mathcal Q}\rho_{j,\ell}
      \right).
\]
Let
\[
  a_\ell^{(0)}=\sum_{j\in\mathcal Q}a_{j,\ell},
  \qquad
  \rho_\ell^{(0)}=\sum_{j\in\mathcal Q}\rho_{j,\ell}.
\]
Then \(A_\ell^{(0)}=\Com(a_\ell^{(0)};\rho_\ell^{(0)})\), and
\begin{align*}
  \EvalCom(\mathbf A^{(0)},i)
  &=
  \bigoplus_{\ell=0}^{t-1}
     i^\ell\odot A_\ell^{(0)}                                    \\
  &=
  \Com\!\left(
     \sum_{\ell=0}^{t-1}a_\ell^{(0)}i^\ell;
     \sum_{\ell=0}^{t-1}\rho_\ell^{(0)}i^\ell
     \right)                                                      \\
  &=
  \Com(s_i^{(0)};\omega_i^{(0)}).
\end{align*}
This establishes the commitment component of the epoch invariant.

It remains to verify the public key.  Let
\(\mathcal Q=(j_1,\ldots,j_m)\) be the ordering used in the link-proof
chain.  Protocol~1 sets \(P_0=E_0\) and
\[
  P_h=[a_{j_h,0}]P_{h-1}
  \qquad\text{for }1\leq h\leq m.
\]
Repeated application of the group-action law gives
\[
  P_h=
  \left[\sum_{u=1}^{h}a_{j_u,0}\right]E_0.
\]
Consequently,
\[
  \pk=P_m
  =\left[\sum_{j\in\mathcal Q}a_{j,0}\right]E_0
  =[k]E_0.
\]
Every link proof verifies by completeness of the NIZK system.

Finally, each constant term \(a_{j,0}\) is sampled independently and
uniformly from \(\Zq\).  The sum of one or more independent uniform
elements of \(\Zq\) is itself uniform.  Hence, when
\(\mathcal Q\neq\varnothing\), the resulting master key \(k\) is
uniformly distributed in \(\Zq\).
\end{proof}

\subsection{Correctness of threshold evaluation}

The evaluation protocol must reproduce the action of the master key even
though no server individually holds \(k\).

\begin{lemma}[Correctness of threshold evaluation]
\label{lem:eval-correct}
Let epoch \(e\) satisfy Definition~\ref{def:valid-epoch-state}, and let
\(I=\{i_1,\ldots,i_t\}\) be a quorum of \(t\) distinct servers.  Assume
that every server in \(I\) uses its certified epoch-\(e\) share and that
the client follows Protocol~2 honestly.  Then the client obtains
\[
  Y=[k]\Hone(\ctx\|x)
\]
and outputs \(F_k(x)\) as defined in Equation~\ref{eq:prf}.
\end{lemma}

\begin{proof}
Let
\[
  X=\Hone(\ctx\|x)
  \quad\text{and}\quad
  B=[r]X,
\]
where \(r\leftarrow\Zq\) is the client's blinding value.  The servers
construct the evaluation chain
\[
  Q_0=B,
  \qquad
  Q_h=
  [\lambda_{i_h}^{I}s_{i_h}^{(e)}]Q_{h-1}
  \quad\text{for }1\leq h\leq t.
\]

We first prove by induction on \(h\) that
\[
  Q_h=
  \left[
    \sum_{u=1}^{h}
      \lambda_{i_u}^{I}s_{i_u}^{(e)}
  \right]B.
\]
The claim is immediate for \(h=0\), since the empty sum is \(0\) and
\([0]B=B\).  Suppose that it holds for \(h-1\).  Then
\begin{align*}
  Q_h
  &=
  [\lambda_{i_h}^{I}s_{i_h}^{(e)}]Q_{h-1}                         \\
  &=
  [\lambda_{i_h}^{I}s_{i_h}^{(e)}]
  \left[
    \sum_{u=1}^{h-1}
      \lambda_{i_u}^{I}s_{i_u}^{(e)}
  \right]B                                                        \\
  &=
  \left[
    \sum_{u=1}^{h}
      \lambda_{i_u}^{I}s_{i_u}^{(e)}
  \right]B,
\end{align*}
where the last equality follows from
\([a]([b]E)=[a+b]E\).

For \(h=t\), Lagrange interpolation of the sharing polynomial \(F_e\)
at zero gives
\[
  \sum_{i\in I}\lambda_i^I s_i^{(e)}
  =
  \sum_{i\in I}\lambda_i^I F_e(i)
  =
  F_e(0)
  =
  k.
\]
Therefore,
\[
  Q_t=[k]B.
\]
Since \(B=[r]X\), another application of the group-action law yields
\[
  Q_t=[k]([r]X)=[k+r]X.
\]
The client removes the blinding action,
\[
  Y=[-r]Q_t=[-r]([k+r]X)=[k]X.
\]
It consequently computes
\[
  y=
  \Htwo(
    \mathsf{PIVOT-out}\|\ctx\|\pk\|x\|
    \mathsf{enc}([k]\Hone(\ctx\|x))
  )
  =
  F_k(x).
\]

Because all statements and witnesses used by honest parties satisfy the
corresponding proof relations, every blinding proof and evaluation proof
is accepted by NIZK completeness.  Similarly, all signatures verify by
correctness of the signature scheme.  Hence an honest client accepts the
transcript and obtains the claimed output.
\end{proof}

\subsection{Correctness of proactive share refresh}

A proactive refresh must change the sharing polynomial while preserving
its constant term.  Therefore, the individual shares and their public
commitments may change, but the master key and public key must remain
unchanged.

\begin{lemma}[Correctness of proactive refresh]
\label{lem:refresh-correct}
Suppose that the epoch-\(e\) state is valid for threshold \(t\) and
master key \(k\).  Assume that all parties execute Protocol~3 honestly.
Then the resulting epoch-\((e+1)\) state is also valid for threshold
\(t\) and the same master key \(k\).  In particular,
\[
  F_{e+1}(0)=k,
  \qquad
  A_0^{(e+1)}=A_0^{(e)},
  \qquad
  \pk=[k]E_0.
\]
\end{lemma}

\begin{proof}
For each qualified refresh dealer \(S_j\), let
\[
  z_j(X)=\sum_{\ell=1}^{t-1}b_{j,\ell}X^\ell.
\]
By construction, \(z_j\) has degree at most \(t-1\) and satisfies
\(z_j(0)=0\).  Define the aggregate refresh polynomial
\[
  Z_e(X)=\sum_{j\in\mathcal R_e}z_j(X)
\]
and the next sharing polynomial
\[
  F_{e+1}(X)=F_e(X)+Z_e(X).
\]
Both \(F_e\) and \(Z_e\) have degree at most \(t-1\), so
\(F_{e+1}\) also has degree at most \(t-1\).  Its constant term is
\[
  F_{e+1}(0)
  =
  F_e(0)+Z_e(0)
  =
  k+\sum_{j\in\mathcal R_e}z_j(0)
  =
  k.
\]

For every server \(S_i\), Protocol~3 computes
\begin{align*}
  s_i^{(e+1)}
  &=
  s_i^{(e)}
  +
  \sum_{j\in\mathcal R_e}\delta_{j,i}                              \\
  &=
  F_e(i)
  +
  \sum_{j\in\mathcal R_e}z_j(i)                                   \\
  &=
  F_{e+1}(i).
\end{align*}
Thus the refreshed values form a valid sharing of the same secret.

We next verify the public commitments.  Write
\[
  F_e(X)=\sum_{\ell=0}^{t-1}a_\ell^{(e)}X^\ell.
\]
Then
\[
  F_{e+1}(X)
  =
  a_0^{(e)}
  +
  \sum_{\ell=1}^{t-1}
  \left(
    a_\ell^{(e)}
    +
    \sum_{j\in\mathcal R_e}b_{j,\ell}
  \right)X^\ell.
\]
Hence
\[
  a_0^{(e+1)}=a_0^{(e)}=k
\]
and, for \(1\leq\ell<t\),
\[
  a_\ell^{(e+1)}
  =
  a_\ell^{(e)}
  +
  \sum_{j\in\mathcal R_e}b_{j,\ell}.
\]
By the additive homomorphism of the commitment scheme,
\begin{align*}
  A_\ell^{(e+1)}
  &=
  A_\ell^{(e)}
  \oplus
  \bigoplus_{j\in\mathcal R_e}D_{j,\ell}                           \\
  &=
  \Com\!\left(
    \left(a_\ell^{(e)}
    +
    \sum_{j\in\mathcal R_e}b_{j,\ell}\right)\text{ };\text{ }\left(
    \rho_\ell^{(e)}
    +
    \sum_{j\in\mathcal R_e}\eta_{j,\ell}\right)
  \right)
\end{align*}
for \(1\leq\ell<t\), while
\[
  A_0^{(e+1)}=A_0^{(e)}
\]
because no refresh dealer contributes a constant coefficient.

Let
\[
  \rho_0^{(e+1)}=\rho_0^{(e)}
\]
and, for \(1\leq\ell<t\),
\[
  \rho_\ell^{(e+1)}
  =
  \rho_\ell^{(e)}
  +
  \sum_{j\in\mathcal R_e}\eta_{j,\ell}.
\]
Then \(A_\ell^{(e+1)}
=\Com(a_\ell^{(e+1)};\rho_\ell^{(e+1)})\), and therefore
\[
  \EvalCom(\mathbf A^{(e+1)},i)
  =
  \Com(s_i^{(e+1)};\omega_i^{(e+1)})
\]
for every server \(S_i\).

Finally, the public key depends only on the constant term,
\[
  \pk=[F_{e+1}(0)]E_0=[k]E_0.
\]
Thus all components of the epoch invariant are preserved.
\end{proof}

\subsection{Correctness of the epoch transition}

The algebraic refresh procedure and the activation procedure serve
different purposes.  The refresh procedure constructs a correct next
sharing, whereas the transition procedure determines when that sharing
may be used.

\begin{lemma}[Epoch-consistent activation]
\label{lem:epoch-activation}
Assume an honest execution of the PROPOSE, PREPARE, COMMIT, and ACTIVATE
phases.  Suppose additionally that every evaluation request is bound to
an epoch certificate and that each server applies a share only when the
request epoch and certificate hash match its active local state.  Then
every successful evaluation chain uses shares from one common epoch.
\end{lemma}

\begin{proof}
An evaluation request contains an epoch number \(e\) and is bound to the
corresponding certificate \(\cert_e\).  Before acting, each server checks
that \(e\) equals its active epoch and that the certificate hash equals
the hash stored in its local active state.  Therefore, a server that has
already activated epoch \(e+1\) will not apply an epoch-\(e\) share, and
a server that remains in epoch \(e\) will not apply an epoch-\((e+1)\)
share.

Consequently, a chain can complete only if every participating server
accepts the same epoch identifier and certificate.  Otherwise the
evaluation aborts or is retried after the transition.  Hence every
successful chain is epoch-consistent.
\end{proof}

\begin{remark}[Asynchronous activation]
Receiving \(t\) COMMIT messages does not, by itself, imply that all
honest servers activate the next epoch simultaneously in an asynchronous
network.  The correctness claim needed here is therefore not simultaneous
activation, but epoch-consistent evaluation, a successful evaluation
must use one certificate and one sharing polynomial throughout the
chain.  If the intended model requires uninterrupted availability during
asynchronous transitions, the protocol should additionally specify a
grace-period, final-certificate, or synchronized-activation mechanism.
\end{remark}

\subsection{Correctness of committee resharing}

Committee resharing transfers the sharing of \(k\) from an old committee
with threshold \(t\) to a new committee with threshold \(t'\).

\begin{lemma}[Correctness of committee resharing]
\label{lem:reshare-correct}
Let \(I\) be an old-committee quorum of size at least \(t\), and suppose
that the old epoch state is valid with sharing polynomial \(F_e\) and
master key \(k\).  Assume that all old and new committee members execute
Protocol~5 honestly.  Then the new committee obtains evaluations of a
polynomial \(G(X)\) of degree at most \(t'-1\) satisfying
\[
  G(0)=k.
\]
The new shares and the aggregate resharing commitments are mutually
consistent.
\end{lemma}

\begin{proof}
Each old server \(S_i\), for \(i\in I\), samples a polynomial
\[
  g_i(X)
  =
  \lambda_i^I s_i^{(e)}
  +
  \sum_{\ell=1}^{t'-1}c_{i,\ell}X^\ell.
\]
This polynomial has degree at most \(t'-1\) and constant term
\[
  g_i(0)=\lambda_i^I s_i^{(e)}.
\]
Define
\[
  G(X)=\sum_{i\in I}g_i(X).
\]
Since each summand has degree at most \(t'-1\), so does \(G\).  Its
constant term is
\begin{align*}
  G(0)
  &=
  \sum_{i\in I}g_i(0)                                             \\
  &=
  \sum_{i\in I}\lambda_i^I s_i^{(e)}                              \\
  &=
  \sum_{i\in I}\lambda_i^I F_e(i)                                 \\
  &=
  F_e(0)                                                          \\
  &=
  k.
\end{align*}
The fourth equality is the Lagrange interpolation formula evaluated at
zero.

Each new server \(S_j'\) receives \(g_i(j)\) from every old server
\(S_i\in I\) and computes
\[
  s_j'
  =
  \sum_{i\in I}g_i(j)
  =
  G(j).
\]
Therefore, the values
\(\{s_j'\}_{j=1}^{n'}\) form a threshold-\(t'\) Shamir sharing of the
same master key \(k\).

For commitment consistency, let
\[
  B_{i,\ell}=\Com(c_{i,\ell};\rho_{i,\ell})
\]
for \(\ell\geq 1\), and let \(B_{i,0}\) commit to
\(\lambda_i^I s_i^{(e)}\).  The reshare proof establishes that this
constant-term commitment is linked to the certified old share.  By
homomorphic evaluation of the coefficient commitments,
\[
  \EvalCom(
    (B_{i,0},\ldots,B_{i,t'-1}),j
  )
  =
  \Com(g_i(j);\omega_{i,j})
\]
for the corresponding opening randomness \(\omega_{i,j}\).  Summing
over all \(i\in I\) gives
\[
  \bigoplus_{i\in I}
  \EvalCom(
    (B_{i,0},\ldots,B_{i,t'-1}),j
  )
  =
  \Com\!\left(
    \sum_{i\in I}g_i(j);
    \sum_{i\in I}\omega_{i,j}
  \right)
  =
  \Com(s_j';\omega_j').
\]
Hence the new shares are consistent with the aggregate resharing
commitments.
\end{proof}

\begin{remark}
The statement that no party learns \(k\) during resharing is a privacy
claim rather than a correctness claim.  It should therefore be proved in
the security section under the relevant corruption threshold and erasure
assumptions.  Correctness establishes only that the protocol computes a
new sharing whose constant term is \(k\).
\end{remark}

\subsection{Overall correctness}

We can now combine the preceding lemmas into a global correctness
theorem.

\begin{theorem}[Overall correctness of \scheme]
\label{thm:overall-correct}
Assume that all parties execute the protocol honestly, that every
successful evaluation is epoch-consistent, and that the primitive
correctness assumptions stated at the beginning of this section hold.
Let Protocol~1 generate the initial master key \(k\), public key
\(\pk\), and server shares.  Then the following statements hold,

\begin{enumerate}[label=(\roman*)]
  \item For every valid epoch \(e\), every input
        \(x\in\{0,1\}^*\), and every epoch-\(e\) quorum \(I\) of size
        \(t\), Protocol~2 outputs
        \[
          y=F_k(x).
        \]

  \item After any finite sequence of proactive refresh operations, the
        active shares form a valid threshold sharing of the same master
        key \(k\), and the public key remains
        \[
          \pk=[k]E_0.
        \]

  \item After any valid committee resharing from threshold \(t\) to
        threshold \(t'\), the new committee holds a valid
        threshold-\(t'\) sharing of the same master key \(k\).

  \item For every fixed input \(x\), the value \(F_k(x)\) is invariant
        across all refreshes and resharings.
\end{enumerate}
\end{theorem}

\begin{proof}
The proof proceeds by induction over the sequence of state-changing
protocol executions.

\paragraph{Initial state.}
By Lemma~\ref{lem:dkg-correct}, Protocol~1 establishes a valid
epoch-\(0\) state for a uniformly distributed master key \(k\), with
public key \(\pk=[k]E_0\).

\paragraph{Evaluation.}
Assume the currently active state is valid.  By
Lemma~\ref{lem:epoch-activation}, every successful evaluation uses a
single epoch certificate and shares from one common sharing polynomial.
Lemma~\ref{lem:eval-correct} then implies that the client recovers
\[
  [k]\Hone(\ctx\|x)
\]
and outputs \(F_k(x)\).

\paragraph{Refresh step.}
Assume epoch \(e\) is valid.  Lemma~\ref{lem:refresh-correct} shows that
Protocol~3 produces a valid epoch-\((e+1)\) state whose sharing
polynomial has the same constant term \(k\).  The public key therefore
remains \([k]E_0\).  This proves preservation of the invariant across
every proactive refresh.

\paragraph{Resharing step.}
Assume the old committee holds a valid sharing of \(k\).  By
Lemma~\ref{lem:reshare-correct}, Protocol~5 constructs a new sharing
polynomial \(G\) with \(G(0)=k\).  Thus the new committee satisfies the
same invariant, with threshold \(t'\) and its corresponding commitment
vector.

\paragraph{Output invariance.}
The OPRF value is
\[
  F_k(x)
  =
  \Htwo(
    \mathsf{PIVOT-out}\|\ctx\|\pk\|x\|
    \mathsf{enc}([k]\Hone(\ctx\|x))
  ).
\]
Neither proactive refresh nor committee resharing changes \(k\),
\(\pk=[k]E_0\), \(\ctx\), or \(x\).  Therefore the input to \(\Htwo\)
is unchanged, and so the resulting OPRF output is identical throughout
the lifetime of the protocol.
\end{proof}

\begin{corollary}[Correctness after arbitrary maintenance operations]
\label{cor:maintenance-correct}
Let the system undergo any finite sequence consisting of proactive
refreshes and valid committee resharings.  If the resulting active
committee completes an epoch-consistent evaluation on input \(x\), then
the client obtains the same value \(F_k(x)\) that it would have obtained
immediately after the original DKG.
\end{corollary}

\begin{proof}
Each refresh and resharing operation preserves the master key \(k\) by
Theorem~\ref{thm:overall-correct}.  The conclusion therefore follows
from the correctness of threshold evaluation.
\end{proof}

\section{Security Analysis}
\label{sec:security}

This section proves the security of \scheme\ against a mobile
\emph{semi-honest} (honest-but-curious) adversary.  A corrupted party
executes the prescribed algorithms faithfully, uses correctly distributed
randomness, sends all required messages, and does not abort or modify a
message.  The adversary may nevertheless retain and jointly analyse the
complete internal states and transcripts of the parties that it corrupts.
Consequently, the proof in this section addresses confidentiality and
privacy. Robustness, blame, and resistance to malformed proofs belong to
the malicious-security analysis.

\subsection{Adversarial model and ideal functionality}
\label{sec:sh-model}

Time is divided into epochs.  In epoch $e$, the adversary corrupts a set
$\mathcal C_e\subseteq[n]$ satisfying $|\mathcal C_e|<t$.  The set may
change between epochs.  Corruptions are \emph{epoch respecting} and after an
epoch transition, an honest server erases its epoch-$e$ share and refresh
randomness before a corruption in epoch $e+1$ reveals its state.  Thus, a
mobile adversary never obtains both the erased and current states of an
honest server.  This condition is necessary for every proactive secret
sharing protocol, without it, an adversary could sequentially collect $t$
shares from one epoch and reconstruct the key.

Clients may also be corrupted.  The corruption status of a client is fixed
before the beginning of each evaluation session and an honest client remains
uncorrupted for that session, whereas a corrupted client is passive from the
start.  This standard static-per-session restriction avoids the stronger
problem of explaining an already transmitted blinded curve after a later
client corruption. Supporting such post-session adaptive client corruption
would require a non-committing simulation mechanism or an explicit
simulator-only input interface.  A corrupted client forms
$B=[r]\Hone(\ctx\|x)$ with a fresh uniform $r$, submits a valid blinding
proof, verifies the server proofs, and computes the specified output.  Network
metadata including $\sid$, epoch number, quorum identity, message lengths,
and success of a session is treated as public leakage.

We compare the real execution with an ideal functionality
$\F^{\mathsf{sh}}_{\mathsf{pTVOPRF}}$ having the following external
behaviour.

\begin{itemize}[nosep]
  \item On $\mathsf{Setup}(n,t)$, it samples $k\leftarrow\Zq$, forms a
        uniformly random degree-$(t-1)$ Shamir sharing of $k$, and publishes
        $\pk=[k]E_0$.
  \item On $\mathsf{Eval}(x)$, it returns $F_k(x)$ to the requesting client
        and leaks only the public session metadata to the adversary.
  \item On $\mathsf{Refresh}$, it replaces the current sharing by an
        independently random degree-$(t-1)$ sharing of the same $k$ and
        erases the old sharing.
  \item On $\mathsf{Reshare}(n',t')$, it replaces the old sharing by an
        independently random degree-$(t'-1)$ sharing of the same $k$ among
        the new committee and erases the old-committee state.
  \item On corruption of server $S_i$ in epoch $e$, it reveals only the
        current share $s_i^{(e)}$ and the current local state.  On corruption
        of a client, it reveals that client's input, randomness, intermediate
        value $Y=[k]\Hone(\ctx\|x)$, and output, exactly as a real passive
        corruption would.
\end{itemize}

For proof convenience, the simulator and the ideal functionality jointly
sample the same uniform key and Shamir shares, the simulator never releases the
key to the adversary.  Equivalently, one may regard this as a private
simulation interface that does not change the external input/output behaviour
of the functionality.

\begin{definition}[Semi-honest realization]
Protocol \scheme\ securely realizes
$\F^{\mathsf{sh}}_{\mathsf{pTVOPRF}}$ if, for every PPT semi-honest
adversary $\A$ corrupting fewer than $t$ servers in each epoch, there exists
a PPT simulator $\Sc$ such that the joint distribution of the environment's
output and the adversary's view in the real execution is computationally
indistinguishable from the corresponding distribution in the ideal
execution with $\Sc$.
\end{definition}

\subsection{Elementary privacy lemmas}
\label{sec:sh-lemmas}

\begin{lemma}[Privacy of fewer than $t$ Shamir shares]
\label{lem:shamir-passive}
Let $f(X)=k+\sum_{j=1}^{t-1}a_jX^j$, where
$a_1,\ldots,a_{t-1}\leftarrow\Zq$ independently.  For every set
$C=\{i_1,\ldots,i_c\}$ with $c<t$, the vector
$(f(i_1),\ldots,f(i_c))$ is uniformly distributed over $\Zq^c$ and is
independent of $k$.
\end{lemma}

\begin{proof}
Write
\[
  \begin{pmatrix}f(i_1)\\ \vdots\\ f(i_c)\end{pmatrix}
  = k\mathbf 1 +
  \begin{pmatrix}
  i_1&i_1^2&\cdots&i_1^{t-1}\\
  \vdots&\vdots&&\vdots\\
  i_c&i_c^2&\cdots&i_c^{t-1}
  \end{pmatrix}
  \begin{pmatrix}a_1\\ \vdots\\ a_{t-1}\end{pmatrix}.
\]
The displayed Vandermonde submatrix has row rank $c$ because the indices are
distinct and non-zero in $\Zq$.  Hence the associated linear map from
$\Zq^{t-1}$ to $\Zq^c$ is surjective.  A uniform coefficient vector is
therefore mapped to a uniform vector in $\Zq^c$.  Adding the fixed vector
$k\mathbf 1$ only translates the uniform distribution and does not change
it.  Thus fewer than $t$ shares contain no information about $k$.
\end{proof}

\begin{lemma}[Perfect privacy of the blinded input]
\label{lem:blind-perfect}
For every two inputs $x_0,x_1$, the curve-valued portions of an honest
client's evaluation transcript have identical distributions.  After
including $\pi^{\mathsf{blind}}$, the complete request transcripts are
computationally indistinguishable under the zero-knowledge property of the
blinding NIZK.
\end{lemma}

\begin{proof}
Fix $x$ and write $X=\Hone(\ctx\|x)$.  Since the action is free and
transitive and $|G|=|\E|=q$, the map
\[
  \phi_X:\Zq\longrightarrow\E,\qquad r\longmapsto[r]X
\]
is a bijection.  Therefore, for uniform $r\leftarrow\Zq$, the blinded
curve $B=[r]X$ is uniform over $\E$, independently of $X$ and hence
independently of $x$.

For a fixed quorum $I=(i_1,\ldots,i_t)$, define
$\alpha_h=\sum_{j=1}^{h}\lambda_{i_j}^I s_{i_j}^{(e)}$.  The curve chain is
\[
  (Q_0,Q_1,\ldots,Q_t)
  =(B,[\alpha_1]B,\ldots,[\alpha_t]B).
\]
It is a deterministic function of the uniform curve $B$ and values that do
not depend on $x$.  Consequently, its joint distribution is identical for
$x_0$ and $x_1$, this is stronger than equality of the individual marginal
distributions.

The evaluation proofs and signatures are generated from the curve statements
and the servers' shares, all of which are independent of $x$ once $B$ is
fixed.  Finally, zero knowledge permits replacement of the honest client's
blinding proof by a simulated proof without revealing the witness $(x,r)$.
Thus the complete server-side view is computationally independent of the
client input.
\end{proof}

\begin{lemma}[Refresh re-randomizes the sharing]
\label{lem:refresh-randomizes}
Assume that at least one refresh dealer is honest.  Conditioned on the fixed
master key $k$, the refreshed polynomial $F_{e+1}$ is uniformly distributed
over
\[
  \mathcal P_k=\{f\in\Zq[X]:\deg(f)<t\text{ and }f(0)=k\},
\]
and is independent of $F_e$ from the adversary's perspective.
\end{lemma}

\begin{proof}
Let
\[
  \mathcal Z_0=\{z\in\Zq[X]:\deg(z)<t\text{ and }z(0)=0\}.
\]
An honest refresh dealer samples its non-constant coefficients uniformly,
so its polynomial is uniform in the additive group $\mathcal Z_0$.  The sum
of this uniform polynomial and any fixed or adversarially known collection
of other zero-polynomials is still uniform in $\mathcal Z_0$.  Therefore the
aggregate refresh polynomial $Z_e=\sum_j z_j$ is uniform in $\mathcal Z_0$
and independent of $F_e$.  Since translation by $F_e$ is a bijection from
$\mathcal Z_0$ to $\mathcal P_k$, the polynomial
$F_{e+1}=F_e+Z_e$ is uniform in $\mathcal P_k$.

A corrupted set of size $c<t$ receives only $c$ evaluations of each honest
zero-polynomial.  The same rank argument as in
Lemma~\ref{lem:shamir-passive}, now with the constant term fixed to zero,
shows that these evaluations are uniform over $\Zq^c$.  The broadcast
coefficient commitments reveal no additional information by statistical
hiding.  After activation, secure erasure removes $F_e$-shares and refresh
randomness, hence a later corruption cannot link the two independently
random sharings through erased state.
\end{proof}

\begin{lemma}[Privacy of committee resharing]
\label{lem:reshare-passive}
Suppose fewer than $t$ old servers and fewer than $t'$ new servers are
corrupted.  Protocol~5 gives the new committee a uniformly random
degree-$(t'-1)$ Shamir sharing of the same key $k$ and reveals no additional
information about the honest old shares.
\end{lemma}

\begin{proof}
For each old server $S_i\in I$, the resharing polynomial is
\[
  g_i(X)=\lambda_i^Is_i^{(e)}+
         \sum_{\ell=1}^{t'-1}c_{i,\ell}X^\ell.
\]
Consider an honest old server $S_i$ and a corrupted new-server set
$C'$ of size $c'<t'$.  By the Vandermonde-rank argument, the vector
$(g_i(j))_{j\in C'}$ is uniform over $\Zq^{c'}$ independently of the
constant term $\lambda_i^Is_i^{(e)}$.  Thus the individual messages sent by
an honest old dealer to corrupted new servers hide that dealer's old share.
The coefficient commitments are statistically hiding, and the reshare NIZK
is zero knowledge.

The aggregate polynomial $G(X)=\sum_{i\in I}g_i(X)$ has
\[
  G(0)=\sum_{i\in I}\lambda_i^Is_i^{(e)}=k.
\]
At least one old server in $I$ is honest because $|I|\ge t$ and fewer than
$t$ old servers are corrupted.  The higher coefficients contributed by that
honest dealer are uniform and adding the remaining dealers' coefficients leaves
the aggregate higher-coefficient vector uniform in $\Zq^{t'-1}$.  Hence
$G$ is a uniformly random element of the set of degree-$(t'-1)$ polynomials
with constant term $k$.  Erasure of the old committee's state prevents later
combination of obsolete and current local states.
\end{proof}

\begin{theorem}[Semi-honest security of \scheme]
\label{thm:semi-honest}
Assume that the commitment scheme is statistically hiding, the four NIZKs
are computationally zero knowledge, the isogeny action is free and
transitive, and honest servers securely erase obsolete state.  Then
\scheme\ securely realizes
$\F^{\mathsf{sh}}_{\mathsf{pTVOPRF}}$ against every PPT mobile
semi-honest adversary corrupting fewer than $t$ servers in each epoch and,
during committee migration, fewer than $t'$ new servers, with client
corruption fixed at the start of each evaluation session.

The simulation error is bounded by the sum of the zero-knowledge advantages
of the simulated NIZK proofs, the hiding advantages of any commitments that
are replaced in the simulation, and a negligible secure-erasure failure
probability.  In particular, simulation-extractability, signature
unforgeability, and complaint soundness are not needed for this passive
result.
\end{theorem}

\begin{proof}
We construct a simulator $\Sc$ that runs an internal copy of $\A$ and
simulates all messages sent by honest parties.  The simulator maintains a
uniform master key $k$, the current Shamir polynomial $F_e$, all current
shares and commitment openings, and the random-oracle tables.  This state is
private to $\Sc$ and is never released except to the extent prescribed by a
passive corruption.

\paragraph{Simulation of the DKG.}
For every corrupted dealer, $\Sc$ lets the internal adversary sample the
dealer's coefficients, commitment randomness, and VSS messages exactly as
in the real protocol.  Let $k_{\mathsf{cor}}$ be the sum of the corrupted
dealers' constant terms.  For the honest dealers, $\Sc$ samples their
polynomials uniformly subject only to
\[
  \sum_{j\in\mathsf{honest}}a_{j,0}=k-k_{\mathsf{cor}}.
\]
This is exactly the real conditional distribution.  Indeed, in a real DKG
all dealer constants are independent and uniform, and, conditioned on their
total being $k$ and on the corrupted constants, the honest constants are
uniform over the displayed affine hyperplane.  All non-constant coefficients
remain independent and uniform.

The simulator computes the honest-to-corrupted VSS shares and openings,
publishes honest coefficient commitments, and generates the link proofs and
signatures with the genuine witnesses it knows.  Since all parties are
semi-honest, no complaint occurs.  The aggregate polynomial has constant
term $k$, the public key is exactly $[k]E_0$, and every corrupted server's
share, received subshares, openings, random tape, and public transcript have
the same distribution as in the real DKG.  Honest dealers erase their DKG
polynomials and distribution randomness after the epoch-zero certificate is
finalized, later corruptions therefore reveal only the retained current
share and opening.

\paragraph{Evaluation with an honest client.}
The ideal functionality does not reveal the honest input $x$ to $\Sc$.
The simulator samples $\widehat B\leftarrow\E$ uniformly and produces a
simulated proof
$\widehat\pi^{\mathsf{blind}}$ for the public statement
$(\ctx,\sid,\widehat B)$.  By Lemma~\ref{lem:blind-perfect},
$\widehat B$ has exactly the same distribution as
$[r]\Hone(\ctx\|x)$, and zero knowledge makes the simulated proof
indistinguishable from the real one.

Starting with $\widehat Q_0=\widehat B$, the simulator uses the maintained
shares to compute
\[
  \widehat Q_h=[\lambda_{i_h}^Is_{i_h}^{(e)}]\widehat Q_{h-1}
\]
for every honest and corrupted position in the chain.  It creates genuine
evaluation proofs and signatures because it knows all corresponding shares
and openings.  The resulting chain is distributed exactly as a real chain
conditioned on its first curve.  The ideal functionality returns $F_k(x)$
to the honest client.  The server adversary learns neither $x$ nor the
client's blinding scalar, its complete view is therefore indistinguishable
from real by Lemma~\ref{lem:blind-perfect}.

Notice that no assumption that the simulator learns the honest client's
input is required.  For the actual hidden $X=\Hone(\ctx\|x)$, transitivity
implies that there exists a unique scalar $r$ satisfying
$\widehat B=[r]X$.  Thus the simulated curve chain is algebraically
consistent with some correctly distributed client randomness even though
$\Sc$ does not know that scalar.

\paragraph{Evaluation with a corrupted client.}
A semi-honest corrupted client exposes its input $x$ and randomness $r$ to
the internal adversary and forms the prescribed request.  The simulator
checks the request only as an honest server would, computes all honest-server
partial evaluations from the maintained shares, and produces genuine
proofs and signatures.  At the end,
\[
  Q_t=[k+r]\Hone(\ctx\|x),\qquad
  Y=[-r]Q_t=[k]\Hone(\ctx\|x),
\]
so the corrupted client's complete internal state and output are identical
to those in a real execution.  This argument also covers sessions in which
some servers and the client are corrupted simultaneously.

\paragraph{Simulation of proactive refresh.}
For corrupted refresh dealers, $\Sc$ lets $\A$ generate the prescribed
uniform zero-polynomials.  It samples every honest dealer's zero-polynomial
and commitment randomness exactly as in Protocol~3, sends the corresponding
zero-shares to corrupted receivers, and simulates the public PREPARE,
COMMIT, and ACTIVATE messages.  Because the adversary is passive, all checks
succeed and every honest party signs the same certificate.

The simulator updates
$F_{e+1}=F_e+\sum_jz_j$ and the associated openings.  By
Lemma~\ref{lem:refresh-randomizes}, the resulting sharing is a fresh uniform
sharing of the same key and is independent of the old sharing from the
adversary's perspective.  At activation, $\Sc$ deletes the old shares,
received zero-shares, and refresh randomness.  A corruption in epoch $e+1$
therefore reveals exactly the state prescribed by the ideal functionality
and no erased epoch-$e$ value.

\paragraph{Simulation of committee resharing.}
The simulator executes every corrupted old dealer according to its real
random tape and samples the honest resharing polynomials exactly as in
Protocol~5.  It generates all commitments, private evaluations, proofs, and
signatures, and computes the aggregate new polynomial $G$.  By
Lemma~\ref{lem:reshare-passive}, $G$ is a fresh, uniformly random
threshold-$t'$ sharing of $k$, and the messages received by fewer than $t'$
corrupted new servers hide the constants of the honest old dealers.  After
the new certificate is accepted, $\Sc$ erases the old committee's local
shares, matching the ideal execution.

\paragraph{Adaptive passive corruptions.}
When $\A$ corrupts server $S_i$ in epoch $e$, $\Sc$ reveals the already
sampled current share $s_i^{(e)}$, its opening, current certificate, and all
non-erased local data.  These values were generated before the corruption
with the correct real distribution, so no equivocation is required.  The
epoch-respecting corruption rule and secure erasure guarantee that obsolete
shares and refresh randomness are absent.  Lemma~\ref{lem:shamir-passive}
and Lemma~\ref{lem:refresh-randomizes} show that the collection of states
obtained in different epochs cannot be combined into $t$ points on a single
sharing polynomial.

All simulated values are therefore identically distributed to real values,
except for honest-client blinding proofs (and any optionally simulated
proofs or commitments), whose replacement is indistinguishable by NIZK zero
knowledge (and commitment hiding).  A standard hybrid replacing those
objects one at a time gives
\[
 \left|\Pr[\Real_{\A,\scheme}=1]-
       \Pr[\Ideal_{\Sc,\F^{\mathsf{sh}}_{\mathsf{pTVOPRF}}}=1]\right|
 \le q_{\mathsf{zk}}\Adv_{\mathsf{ZK}}(\lambda)
    +q_{\mathsf{com}}\Adv_{\mathsf{hide}}(\lambda)
    +\negl(\lambda),
\]
where $q_{\mathsf{zk}}$ and $q_{\mathsf{com}}$ are polynomial bounds on
the number of replaced proofs and commitments.  This quantity is negligible,
which completes the simulation proof.
\end{proof}

The realization theorem establishes that passive protocol transcripts reveal
no more than the ideal leakage.  The following corollaries state the main
cryptographic consequences explicitly.

\begin{corollary}[Client-input privacy]
For any two equal-length inputs $x_0,x_1$, any PPT coalition of fewer than
$t$ semi-honest servers has negligible advantage in distinguishing an
evaluation of $x_0$ from an evaluation of $x_1$.
\end{corollary}

\begin{proof}
The curve transcript is identically distributed by
Lemma~\ref{lem:blind-perfect}, only the zero-knowledge replacement of the
blinding proof contributes a negligible distinguishing term.
\end{proof}

\begin{corollary}[Master-key privacy]
A coalition of fewer than $t$ semi-honest servers learns no information
about $k$ from its shares.  Recovering $k$ from the public key or from the
public evaluation pairs is no easier than solving vectorization in the
underlying group action.
\end{corollary}

\begin{proof}
Information-theoretic secrecy of the share vector follows from
Lemma~\ref{lem:shamir-passive}.  Statistical hiding protects the committed
coefficients, while zero knowledge protects the witnesses in the link,
evaluation, and resharing proofs.  The public key is the vectorization
instance $(E_0,[k]E_0)$.  Moreover, a complete evaluation exposes a pair
$(B,[k]B)$.  Polynomially many such random-base pairs do not make the problem
easier under random self-reducibility, from a challenge $(E_0,[k]E_0)$, a
reduction chooses $u\leftarrow\Zq$ and forms
$B=[u]E_0$ and $[u]([k]E_0)=[k]B$.  Therefore an algorithm recovering $k$
from the public transcript yields an algorithm for vectorization.
\end{proof}

\paragraph{Required one-more assumption.}
A corrupted client learns the unblinded group-action value
\[
  Y_x=[k]\Hone(\ctx\|x)
\]
before hashing it.  Hence a one-more assumption whose oracle returns only
$F_k(x)$ does not fully model the real client view.  For a fresh-output claim,
we use the following stronger and protocol-faithful game.

\begin{definition}[One-more hidden-group-action game]
\label{def:om-hga}
The challenger samples $k\leftarrow\Zq$, publishes
$\pk=[k]E_0$, and gives the adversary oracle access to
\[
  \mathcal O_k^{\mathsf{act}}(x)
  =[k]\Hone(\ctx\|x).
\]
After at most $q$ queries, the adversary outputs $(x^*,Y^*)$ and wins if
$x^*$ was not queried and
$Y^*=[k]\Hone(\ctx\|x^*)$.  The
$\mathsf{OM\text{-}HGA}$ assumption states that every PPT adversary wins
with negligible probability.
\end{definition}

\begin{corollary}[Fresh-output pseudorandomness]
\label{cor:fresh-output}
In the random-oracle model for $\Htwo$, a semi-honest client making at most
$q$ evaluations cannot compute $F_k(x^*)$ for a fresh input $x^*$ except
with probability
\[
  \Adv_{\mathsf{OM\text{-}HGA}}(\lambda)
  +\frac{q_{\Htwo}}{2^{\ell}},
\]
where $q_{\Htwo}$ is the number of its $\Htwo$ queries.
\end{corollary}

\begin{proof}
Suppose an adversary outputs $y^*=F_k(x^*)$ for a fresh $x^*$.  If it never
queries $\Htwo$ at
\[
  \textsf{PIVOT-out}\|\ctx\|\pk\|x^*\|
  \mathsf{enc}([k]\Hone(\ctx\|x^*)),
\]
then $y^*$ is an independent $\ell$-bit random value and is guessed with
probability at most $2^{-\ell}$ per relevant attempt.  Otherwise, the
correct random-oracle query contains
$Y^*=[k]\Hone(\ctx\|x^*)$, extracting that query gives a successful
$\mathsf{OM\text{-}HGA}$ adversary.  A union bound over the
$q_{\Htwo}$ oracle queries gives the stated bound.
\end{proof}

\begin{remark}[Assumptions needed only for malicious security]
In the semi-honest model, parties never produce malformed NIZKs, equivocate
commitments, forge signatures, submit invalid refresh contributions, or
abort strategically.  Consequently, simulation-extractability, commitment
binding, EUF-CMA unforgeability, blame soundness, and the robust-restart
argument are not used in Theorem~\ref{thm:semi-honest}.  They should be
retained for the malicious-security theorem, but invoking them in the
passive proof obscures which assumptions protect privacy and which protect
active correctness.
\end{remark}

\section{Efficiency, Applications, and Conclusion}
\label{sec:effi}

This section examines the communication and computational costs of
\scheme{} at the level of its individual sub-protocols. The purpose of
the analysis is not to claim that the construction is as lightweight as
a conventional single-server OPRF.  Such a comparison would overlook
the additional functionality provided by the protocol.  In addition to
oblivious evaluation, \scheme{} provides dealerless threshold key
generation, public certification of epoch shares, proactive protection
against a mobile adversary, identifiable failures, secure epoch
transitions, and migration to a new committee without changing the
master key.  These properties necessarily introduce communication and
computation that are absent from protocols designed for a shorter-lived
or less demanding security model.

We separate one-time setup costs, per-evaluation online costs, periodic
maintenance costs, and occasional committee-migration costs.  This
distinction is important in practice.  Distributed key generation is
normally executed only once, proactive refresh is performed once per
epoch, and committee resharing is expected to be comparatively rare.
Only threshold evaluation lies on the critical path of every client
request.  Consequently, a high setup or refresh cost may be acceptable
when it is amortized over a large number of evaluations, whereas the
sequential cost of evaluation directly affects client-visible latency.

\paragraph{Cost notation.}
Let \(n\) denote the size of the current committee, \(t\) its threshold,
\(q_{\mathsf{DKG}}=|\mathcal Q|\) the number of qualified DKG dealers,
and \(r_e=|\mathcal R_e|\) the number of qualified refresh dealers in
epoch \(e\).  For committee resharing, let \(m=|I|\) be the size of the
old reconstruction quorum, \(n'\) the number of new servers, and \(t'\)
the new threshold.  In a normal all-honest execution,
\(q_{\mathsf{DKG}}\approx n\), \(r_e\approx n\), and the old resharing
quorum is usually chosen with \(m=t\).

We denote the encoded sizes of a scalar, commitment opening randomness,
commitment, curve, signature, certificate, and relation-specific NIZK
proof by
\[
  \ell_{\mathbb Z},\quad
  \ell_{\mathsf R},\quad
  \ell_{\mathsf C},\quad
  \ell_{\mathsf E},\quad
  \ell_{\mathsf{sig}},\quad
  \ell_{\mathsf{cert}},\quad
  \ell_{\pi}^{\mathsf{link}},
  \ell_{\pi}^{\mathsf{eval}},
  \ell_{\pi}^{\mathsf{blind}},
  \ell_{\pi}^{\mathsf{reshare}},
\]
respectively.  A privately delivered VSS value and its opening
randomness have size
\[
  \ell_{\mathsf{sh}}
  =
  \ell_{\mathbb Z}+\ell_{\mathsf R}.
\]
Small metadata fields, such as protocol tags, indices, counters,
quorum descriptions, and fixed-length hashes, are suppressed in the
asymptotic expressions but must be included in a concrete
implementation.

For computation, let \(T_{\mathsf{act}}\) be the cost of one isogeny
group action, \(T_{\mathsf{com}}\) the cost of generating one
commitment, and \(T_{\mathsf{ec}}(d)\) the cost of evaluating a
coefficient-commitment vector of length \(d\).  We write
\(T_{\mathsf P}^{\mathcal R}\) and
\(T_{\mathsf V}^{\mathcal R}\) for proof generation and verification
for relation \(\mathcal R\), and
\(T_{\mathsf S}\) and \(T_{\mathsf{SV}}\) for signature generation and
verification.  The NIZK costs are left symbolic because they depend
strongly on the eventual proof-system instantiation.  In particular, it
would be misleading to convert every proof directly into a fixed number
of group actions before a concrete proof system for the joint
commitment--isogeny relations has been specified.

\subsection{Communication cost}
\label{sec:communication-cost}

Communication can be measured in two different ways.  The
\emph{logical transcript size} counts each broadcast object once and is
useful for describing the public protocol transcript.  The
\emph{aggregate network traffic} counts the number of point-to-point
deliveries required to disseminate those objects.  Under an ideal
broadcast functionality, these two views are often conflated.  In an
actual network, however, broadcasting one object to \(n-1\) recipients
may require \(n-1\) deliveries or an equivalent multicast service.  We
therefore discuss both views in cases where the distinction changes the asymptotic
cost.

\paragraph{Dealerless distributed key generation.}
During the DKG, each of the \(n\) servers publishes \(t\) commitments to
the coefficients of its polynomial.  The public commitment transcript
therefore contains
\[
  nt\,\ell_{\mathsf C}
\]
bits.  Each dealer also privately sends one polynomial evaluation and
the corresponding opening randomness to every other server.  Ignoring
the dealer's local self-share, this contributes
\[
  n(n-1)\ell_{\mathsf{sh}}
\]
bits of point-to-point communication.

After the VSS qualification phase, every qualified dealer contributes
one curve and one link proof to the public-key chain.  This adds
\[
  q_{\mathsf{DKG}}
  \bigl(
    \ell_{\mathsf E}
    +
    \ell_{\pi}^{\mathsf{link}}
  \bigr)
\]
bits to the logical transcript.  Finally, the epoch-zero certificate and
its server signatures contribute approximately
\[
  \ell_{\mathsf{cert}}
  +
  n\ell_{\mathsf{sig}}.
\]
Thus, in a complaint-free execution, the logical DKG payload is
\begin{align}
  \mathsf{Comm}_{\mathsf{DKG}}
  &=
  nt\,\ell_{\mathsf C}
  +
  n(n-1)\ell_{\mathsf{sh}}
  +
  q_{\mathsf{DKG}}
  \bigl(
    \ell_{\mathsf E}
    +
    \ell_{\pi}^{\mathsf{link}}
  \bigr)
  +
  n\ell_{\mathsf{sig}}
  +
  \ell_{\mathsf{cert}}.
  \label{eq:comm-dkg}
\end{align}
The private-share term is quadratic in \(n\), as is usual for
dealerless VSS in which every server acts as a dealer.  If broadcast is
implemented by independent delivery to all recipients, the
coefficient-commitment and link-proof terms also acquire an additional
factor of approximately \(n\) in aggregate network traffic.

The DKG is not constant-round in the current specification.  Polynomial
commitment, share delivery, and complaint resolution require a small
number of VSS phases, but the public-key link chain contains
\(q_{\mathsf{DKG}}\) sequential actions because the next curve depends
on the preceding curve.  Consequently, the latency of the DKG contains
an \(O(q_{\mathsf{DKG}})\) sequential component even though most VSS
messages can be sent in parallel.  This is acceptable for a one-time
initialization procedure, but it should be stated explicitly.

\paragraph{Threshold evaluation.}
A client request contains the blinded curve, the blinding proof, and
session metadata.  Its principal payload is
\[
  \ell_{\mathsf E}
  +
  \ell_{\pi}^{\mathsf{blind}}.
\]
The final verifiable response contains one record for every server in
the quorum.  A record consists, up to fixed metadata, of the output curve
of the partial action, one evaluation proof, and one signature.  Define
\[
  \ell_{\mathsf{rec}}
  =
  \ell_{\mathsf E}
  +
  \ell_{\pi}^{\mathsf{eval}}
  +
  \ell_{\mathsf{sig}}.
\]
If the client already stores the active epoch certificate, the
client-visible communication for one evaluation is approximately
\begin{equation}
  \mathsf{Comm}_{\mathsf{client}}
  =
  \ell_{\mathsf E}
  +
  \ell_{\pi}^{\mathsf{blind}}
  +
  t\,\ell_{\mathsf{rec}}.
  \label{eq:comm-eval-client}
\end{equation}
If the certificate is transmitted with every response, an additional
\(\ell_{\mathsf{cert}}\) bits are required.  The client transcript is
therefore linear in the threshold:
\[
  \mathsf{Comm}_{\mathsf{client}}=O(t).
\]

There is an important distinction between this client transcript and
the aggregate traffic generated by the literal forwarding rule in
Protocol~2.  At position \(h\), the server forwards the complete prefix
of \(h\) evaluation records so that the next server can verify all
preceding contributions.  The number of transmitted records is then
\[
  1+2+\cdots+t
  =
  \frac{t(t+1)}{2}.
\]
Accordingly, the aggregate evaluation traffic is approximately
\begin{equation}
  \mathsf{Comm}_{\mathsf{eval}}^{\mathsf{network}}
  =
  \ell_{\mathsf E}
  +
  \ell_{\pi}^{\mathsf{blind}}
  +
  \frac{t(t+1)}{2}\,
  \ell_{\mathsf{rec}}
  +
  \ell_{\mathsf{cert}},
  \label{eq:comm-eval-network}
\end{equation}
which is \(O(t^2)\) under the protocol as presently written.  This does
not contradict the \(O(t)\) client-transcript claim: the quadratic term
arises because earlier records are retransmitted across multiple
server-to-server hops.

An implementation may reduce this aggregate traffic by storing the
append-only transcript on an authenticated bulletin board, by forwarding
only the newly created record together with a transcript hash, or by
using a reliable multicast channel.  Such an optimization reduces the
physical traffic toward \(O(t)\), although its effect on the verification
and robustness argument must be specified carefully.  The evaluation
still requires \(t\) sequential server actions, because
\(Q_h\) depends on \(Q_{h-1}\).

\paragraph{Proactive share refresh.}
During refresh, every qualified refresh dealer publishes \(t-1\)
commitments to its nonconstant zero-polynomial coefficients and sends
one zero-share/opening pair to every other server.  The complaint-free
payload is therefore
\begin{align}
  \mathsf{Comm}_{\mathsf{refresh}}
  &=
  r_e(t-1)\ell_{\mathsf C}
  +
  r_e(n-1)\ell_{\mathsf{sh}}
  + \ell_{\mathsf{cert}}
  +
  O(n\ell_{\mathsf{sig}}).
  \label{eq:comm-refresh}
\end{align}
The first line covers polynomial distribution.  The second line covers
the proposed next-epoch certificate and the signed PREPARE, COMMIT, and
activation or certificate messages.  With \(r_e\approx n\), the
point-to-point zero-share distribution is \(O(n^2)\).  If each
transition signature is independently broadcast to all servers, the
aggregate delivery cost of the transition messages is also quadratic in
\(n\), although their logical transcript contains only \(O(n)\)
signatures. Unlike evaluation, refresh has no inherently sequential chain of
\(n\) group-action messages.  Its main phases, that are, commitment broadcast,
private zero-share delivery, verification, certificate proposal,
PREPARE, COMMIT, and activation form a constant number of communication
phases in a complaint-free synchronous execution.  Complaint resolution
or leader replacement may add further phases.  Refresh can normally be
performed outside the critical path of client requests and amortized
over all evaluations completed in the epoch.

\paragraph{Verification, blame, and restart.}
An evaluation blame certificate contains the adjacent curves involved
in the failed transition, the offending evaluation proof, the server
signature, and session metadata.  Its principal size is
\[
  2\ell_{\mathsf E}
  +
  \ell_{\pi}^{\mathsf{eval}}
  +
  \ell_{\mathsf{sig}}.
\]
A refresh complaint is smaller because the coefficient commitments are
already public and it principally contains the disputed share-opening pair,
the dealer signature, and references to the public commitments.  Blame
communication is exceptional rather than part of the honest-case cost. A failed evaluation may be restarted with a new quorum.  If \(f\)
servers are identified and removed before a successful attempt, the
communication is approximately \((f+1)\) times the honest evaluation
cost, plus \(f\) blame certificates.  The robust wrapper permits at most
\(n-t+1\) failed attempts before fewer than \(t\) eligible servers
remain.  This worst-case bound is intentionally conservative and should
not be confused with the normal per-evaluation cost.

\paragraph{Committee resharing.}
Let \(I\) contain \(m\) old servers.  Every old server publishes \(t'\)
commitments to the coefficients of its new sharing polynomial, produces
one resharing proof, and sends one share-opening pair to each of the
\(n'\) new servers.  The resulting payload is
\begin{align}
  \mathsf{Comm}_{\mathsf{reshare}}
  &=
  mt'\ell_{\mathsf C}
  +
  mn'\ell_{\mathsf{sh}}
  +
  m\ell_{\pi}^{\mathsf{reshare}}+
  n'\ell_{\mathsf{sig}}
  +
  \ell_{\mathsf{cert}}.
  \label{eq:comm-reshare}
\end{align}
For the common choice \(m=t\), this becomes
\[
  O\bigl(
    tt'\ell_{\mathsf C}
    +
    tn'\ell_{\mathsf{sh}}
    +
    t\ell_{\pi}^{\mathsf{reshare}}
  \bigr).
\]
Resharing is therefore more expensive than one evaluation but is
expected to occur only when the committee or threshold changes.  The
protocol avoids the considerably larger application-level cost of
generating a new OPRF key and recomputing or re-encrypting all data
derived from the old key.

\begin{table}[H]
\centering
\caption{Dominant communication terms in a complaint-free execution.
Broadcast objects are counted once in the logical transcript.  The
literal cumulative forwarding rule of Protocol~2 produces
\(O(t^2)\) aggregate network traffic even though the final client
transcript is \(O(t)\).}
\label{tab:communication-cost}
\small
\begin{tabular}{lll}
\toprule
Sub-protocol & Dominant payload & Main scaling term \\
\midrule
DKG
&
\(nt\) commitments,
\(n(n-1)\) private shares,
\(q_{\mathsf{DKG}}\) link records
&
\(O(n^2+nt)\)
\\
Evaluation: client view
&
one blinded request and \(t\) evaluation records
&
\(O(t)\)
\\
Evaluation: literal network
&
cumulative prefixes of \(t\) records
&
\(O(t^2)\)
\\
Refresh
&
\(r_e(t-1)\) commitments and
\(r_e(n-1)\) zero-shares
&
\(O(n^2)\) when \(r_e\approx n\)
\\
Resharing
&
\(mt'\) commitments and
\(mn'\) private shares
&
\(O(m(t'+n'))\)
\\
\bottomrule
\end{tabular}
\end{table}

\subsection{Computation cost}
\label{sec:computation-cost}

The computational profile of \scheme{} is heterogeneous.  Polynomial
evaluation and commitment processing dominate the distributed
maintenance protocols, whereas isogeny actions and joint NIZK proofs
dominate online OPRF evaluation.  Because proof generation for an
isogeny-action relation may be substantially more expensive than the
underlying action itself, the analysis keeps action and proof costs
separate.
\paragraph{Dealerless distributed key generation.}
Each server acts simultaneously as one VSS dealer and as a receiver of
the contributions of the other dealers.  In its dealer role, a server
samples \(t\) coefficients and commitment randomness values, generates
\(t\) commitments, and evaluates both its polynomial and its randomness
polynomial at \(n\) server indices.  With straightforward Horner
evaluation, this requires \(O(nt)\) field operations per dealer and
\(O(n^2t)\) field operations over the complete committee.

In its receiver role, each server verifies one contribution from every
dealer.  A verification computes
\(\EvalCom(\mathbf A_j,i)\), which requires \(t\) public scalar
multiplications and \(t-1\) additions in the commitment space, followed
by one commitment to the received share.  Across all dealer--receiver
pairs, the system performs \(O(n^2)\) VSS checks, each involving a
length-\(t\) commitment vector.  The total commitment-processing cost is
therefore \(O(n^2t)\). Every qualified dealer additionally performs one group action and
generates one proof for
\(\mathcal R_{\mathsf{link}}\).  All servers verify the public link
chain.  Ignoring signatures and inexpensive field operations, the
system-wide cryptographic cost can be summarized as
\begin{align}
  \mathsf{Comp}_{\mathsf{DKG}}
  &\approx
  nt\,T_{\mathsf{com}}
  +
  n^2 T_{\mathsf{ec}}(t)
  +
  q_{\mathsf{DKG}}
  \bigl(
    T_{\mathsf{act}}
    +
    T_{\mathsf P}^{\mathcal R_{\mathsf{link}}}
  \bigr)
  +
  nq_{\mathsf{DKG}}
  T_{\mathsf V}^{\mathcal R_{\mathsf{link}}}.
  \label{eq:comp-dkg}
\end{align}
The \(q_{\mathsf{DKG}}\) group actions and link-proof generations lie on
a sequential chain.  Other polynomial and VSS operations can be
parallelized across dealers and receivers.  Since the DKG is executed
once, this relatively high cost is primarily a setup concern.

\paragraph{Threshold evaluation.}
The client computes one hash-to-orbit operation, one group action to
blind the input, and one proof for
\(\mathcal R_{\mathsf{blind}}\).  After receiving the response, it
verifies \(t\) evaluation proofs and \(t\) signatures, performs one
inverse group action to remove the blinding, and evaluates the final
output hash.  Its principal cost is therefore
\begin{align}
  \mathsf{Comp}_{\mathsf{client}}
  &\approx
  T_{\Hone}
  +
  2T_{\mathsf{act}}
  +
  T_{\mathsf P}^{\mathcal R_{\mathsf{blind}}}
  +
  t\left(
    T_{\mathsf V}^{\mathcal R_{\mathsf{eval}}}
    +
    T_{\mathsf{SV}}
  \right)
  +
  T_{\Htwo}.
  \label{eq:comp-eval-client}
\end{align}

Each of the \(t\) servers in the quorum computes its Lagrange
coefficient, performs one partial group action, generates one evaluation
proof, and signs the resulting record.  The basic server-side generation
cost is
\begin{equation}
  t\left(
    T_{\mathsf{act}}
    +
    T_{\mathsf P}^{\mathcal R_{\mathsf{eval}}}
    +
    T_{\mathsf S}
  \right).
  \label{eq:comp-eval-generation}
\end{equation}
Computing each Lagrange coefficient independently requires \(O(t)\)
field operations, giving \(O(t^2)\) field operations across the quorum.
The coefficients may instead be computed together using standard batch
techniques or cached when a quorum is reused. The present protocol requires the server at position \(h\) to verify all
\(h-1\) earlier proofs and signatures.  The total number of intermediate
verification operations is therefore
\[
  \sum_{h=1}^{t}(h-1)
  =
  \frac{t(t-1)}{2}.
\]
Consequently, the literal server-side cost is
\begin{align}
  \mathsf{Comp}_{\mathsf{eval}}^{\mathsf{servers}}
  &\approx
  t\left(
    T_{\mathsf{act}}
    +
    T_{\mathsf P}^{\mathcal R_{\mathsf{eval}}}
    +
    T_{\mathsf S}
  \right)
  +
  \frac{t(t-1)}{2}
  \left(
    T_{\mathsf V}^{\mathcal R_{\mathsf{eval}}}
    +
    T_{\mathsf{SV}}
  \right).
  \label{eq:comp-eval-servers}
\end{align}
The \(t\) action-and-proof generation steps are sequential and determine
online latency.  This is the principal efficiency limitation of the
native threshold design.  Parallel hardware can accelerate the internal
arithmetic of each proof, but it cannot remove the dependency
\(Q_h\leftarrow Q_{h-1}\). The quadratic intermediate-verification term is not intrinsic to the
mathematical OPRF computation, it follows from the conservative rule that
each server rechecks the complete prefix.  A variant in which every
server verifies only its immediate predecessor and the client performs
the final global verification would reduce intermediate verification to
\(O(t)\), but the corresponding blame and robustness argument would
need to be adjusted.  The efficiency claims of the present paper should
therefore use Equation~\eqref{eq:comp-eval-servers} unless such a variant
is formally adopted.

\paragraph{Proactive share refresh.}
Refresh performs no isogeny group action and, in the current protocol,
requires no joint NIZK proof.  Each qualified refresh dealer generates
\(t-1\) commitments and evaluates a degree-at-most-\((t-1)\)
zero-polynomial and its randomness polynomial at \(n\) points.  This
requires
\[
  r_e(t-1)T_{\mathsf{com}}
  +
  O(r_ent)
\]
field operations. Each server verifies one zero-share contribution from every qualified
refresh dealer.  The verification of one contribution evaluates a
length-\((t-1)\) commitment vector and compares it with a commitment to
the received zero-share.  Across the complete committee, this gives
approximately
\[
  nr_e T_{\mathsf{ec}}(t-1)
\]
commitment-evaluation work.  Servers then add the accepted zero-shares
to their local shares and update the \(t-1\) nonconstant public
commitments.  The transition additionally requires signature generation
and verification for the certificate, PREPARE, and COMMIT messages.

A compact system-wide expression is
\begin{align}
  \mathsf{Comp}_{\mathsf{refresh}}
  &\approx
  r_e(t-1)T_{\mathsf{com}}
  +
  nr_e T_{\mathsf{ec}}(t-1)
  +
  O(nT_{\mathsf S})
  +
  O(ntT_{\mathsf{SV}})
  +
  O(r_ent)
  \label{eq:comp-refresh}
\end{align}
field operations.  For \(r_e\approx n\), the refresh cost is quadratic
in the committee size, but it is composed primarily of parallelizable
field and commitment operations.  It does not contain the expensive
sequential isogeny-proof chain that appears in online evaluation.

\paragraph{Verification, blame, and restart.}
A party constructing an evaluation blame certificate has already
performed the failed proof or signature verification.  Producing the
certificate requires only transcript packaging and a signature, while
every observer independently repeats the relevant verification.  A
refresh complaint similarly requires one commitment-consistency check
and signature verification.  These costs occur only when a deviation or
fault is detected.

A retry repeats the complete online evaluation cost with a new quorum.
If \(f\) attempts fail before success, the total computational work is
approximately \((f+1)\) evaluations plus the verification of \(f\)
blame certificates.  This cost is the price of identifiable robustness,
the protocol does not silently accept a malformed response, but obtains
public evidence and continues with a different server set.

\paragraph{Committee resharing.}
Each of the \(m\) old servers samples \(t'-1\) fresh coefficients,
generates \(t'\) commitments, and evaluates its new polynomial at
\(n'\) points.  The old committee therefore performs
\(mt'\) commitment generations and \(O(mn't')\) field operations.
Every old server also generates one proof for
\(\mathcal R_{\mathsf{reshare}}\).

Each new server verifies \(m\) VSS contributions and \(m\) resharing
proofs.  The system-wide cost is approximately
\begin{align}
  \mathsf{Comp}_{\mathsf{reshare}}
  &\approx
  mt'T_{\mathsf{com}}
  +
  mn'T_{\mathsf{ec}}(t')
  +
  mT_{\mathsf P}^{\mathcal R_{\mathsf{reshare}}}
  +
  mn'T_{\mathsf V}^{\mathcal R_{\mathsf{reshare}}}
  +
  O(mn't').
  \label{eq:comp-reshare}
\end{align}
The resharing relation contains commitment consistency but no
isogeny-action assertion.  It may therefore admit a substantially
cheaper proof than
\(\mathcal R_{\mathsf{link}}\) or
\(\mathcal R_{\mathsf{eval}}\), depending on the selected NIZK
instantiation.  Resharing can also be parallelized across the \(m\) old
dealers and \(n'\) new receivers.

\begin{table}[H]
\centering
\caption{Dominant computational operations.  NIZK generation and
verification costs remain symbolic because they depend on the concrete
proof-system instantiation.}
\label{tab:computation-cost}
\begin{tabular}{llll}
\toprule
Sub-protocol
& Group actions
& Proof generation
& Other dominant work \\
\midrule
DKG
&
\(q_{\mathsf{DKG}}\)
&
\(q_{\mathsf{DKG}}\) link proofs
&
\(O(n^2t)\) field/commitment work
\\
Evaluation
&
\(t+2\)
&
one blind proof and \(t\) eval proofs
&
\(O(t^2)\) intermediate verification
\\
Refresh
&
\(0\)
&
\(0\)
&
\(O(nr_et)\) field/commitment work
\\
Resharing
&
\(0\)
&
\(m\) reshare proofs
&
\(O(mn't')\) field/commitment work
\\
\bottomrule
\end{tabular}
\end{table}

\paragraph{Amortized maintenance cost.}
Suppose epoch \(e\) contains \(N_e\) successful OPRF evaluations before
the next refresh.  The refresh overhead amortized over one evaluation is
\[
  \frac{
    \mathsf{Comm}_{\mathsf{refresh}}
  }{N_e}
  \qquad\text{and}\qquad
  \frac{
    \mathsf{Comp}_{\mathsf{refresh}}
  }{N_e}.
\]
For a high-volume service, these quantities can be much smaller than the
online cost of one evaluation.  The DKG cost is amortized over the full
lifetime of the key, while resharing is amortized over the period during
which the new committee remains active.  The dominant recurring
bottleneck is therefore not proactive refresh itself, but the
sequential generation and verification of the joint evaluation proofs.

\subsection{Comparative analysis and interpretation}
\label{sec:comparative-efficiency}

A direct comparison based only on communication size or round count
places \scheme{} at a disadvantage.  A conventional single-server OPRF
can have a constant-size transcript and a constant number of rounds
because one server already possesses the complete key.  A threshold
protocol implemented through an MPC-emulated virtual server may also
hide the internal committee structure from the client and produce a
constant-size client transcript.  In contrast, \scheme{} exposes every
server's contribution so that the client can verify and attribute the
complete evaluation chain.  This design results in \(t\) sequential
partial actions, \(t\) evaluation proofs, and \(t\) signatures.

The additional cost should therefore be interpreted as the price of a
different functionality profile.  The protocol does not merely divide a
static OPRF key.  It supports dealerless creation of the key, public
certification of the current sharing polynomial, proactive renewal
against a mobile adversary, explicit epoch transitions, public blame,
and committee migration while preserving the key.  Several of these
features have no counterpart in a non-threshold or non-proactive OPRF,
and their costs cannot be removed by a more favorable accounting
convention.

\begin{table}[ht]
\centering
\caption{Qualitative efficiency--functionality comparison.  The
asymptotic values describe the interfaces considered in this paper and
do not normalize the concrete costs of different algebraic assumptions
or proof systems.}
\label{tab:comparative-efficiency}
\small
\begin{tabular}{lccccc}
\toprule
Protocol
& Online rounds
& Client transcript
& Verifiable
& Proactive
& PQ \\
\midrule
Jarecki--Liu~\citep{jarecki2009}
& \(O(1)\)
& \(O(1)\)
& No
& No
& No
\\
Baecker et al.~\citep{baecker2025}
& \(O(1)\)
& \(O(1)\)
& No
& Yes
& No
\\
Pedersen~\citep{pedersen2026}
& \(O(1)\)
& \(O(1)\)
& Yes
& No
& Yes
\\
\scheme{}
& \(O(t)\)
& \(O(t)\)
& Yes
& Yes
& Yes
\\
\bottomrule
\end{tabular}
\end{table}

Relative to the classical OPRF of Jarecki and Liu, \scheme{} incurs
threshold coordination, post-quantum group-action costs, and proof
overhead, but removes the single key-holding server and adds proactive
maintenance.  Relative to the proactive threshold OPRF of Baecker et
al., the present construction additionally targets post-quantum security
and verifiable partial evaluation, at the cost of a larger and
sequential online transcript.  Relative to the isogeny-based threshold
VOPRF of Pedersen, \scheme{} gives up the constant-size virtual-server
interface in order to maintain explicit, independently refreshable
Shamir shares and individually attributable server actions.

The most significant online disadvantage is the sequential evaluation
path.  Even if all servers compute quickly, the next action cannot begin
until the preceding curve has been received and verified.  The NIZK
proof for \(\mathcal R_{\mathsf{eval}}\) is also likely to dominate
computation because it must connect a committed scalar to an isogeny
action.  The present construction should therefore be viewed as a
feature-complete protocol framework rather than as a claim of immediate
low-latency deployment. The maintenance costs are more favorable when viewed over the lifetime
of a service.  DKG is a one-time operation.  Refresh is quadratic in the
committee size but contains no isogeny action or joint action proof, can
be parallelized, and is amortized over an epoch.  Resharing is more
expensive than refresh but occurs only when committee membership or the
threshold changes.  In return, both operations preserve
\[
  k,\qquad
  \pk=[k]E_0,\qquad
  F_k(x).
\]
This preservation has an application-level efficiency benefit that is
not visible in the protocol transcript alone, a large encrypted
database indexed by OPRF outputs does not need to be recomputed or
re-encrypted after a refresh or committee migration.

\section{Conclusion}
\label{sec:conclusion}

We have introduced \scheme, the first threshold VOPRF from isogeny
group actions that provides proactive security against a mobile
adversary.  The protocol periodically refreshes all server shares
without changing the master key, public key, or any previously generated
OPRF output, and supports committee migration, publicly verifiable
blame, and coordinated epoch transitions. The construction makes one architectural choice, it exposes the
threshold structure as an explicit sequential chain rather than hiding
it behind an MPC-emulated virtual server.  This makes the
secret-sharing state directly visible and amendable, so proactive
refresh reduces to adding verifiable zero-sharing polynomials, no
isogeny evaluations, no new proofs.  Committee resharing operates on
the same explicit sharing polynomial.

Correctness is established by an invariant maintained across DKG,
evaluation, refresh, and resharing: the master key, public key, and
OPRF outputs are preserved throughout the system lifetime.  A
simulation-based security proof shows that fewer than \(t\) shares per
epoch reveal nothing about the key, that evaluation transcripts are
independent of client inputs, and that cross-epoch share accumulation
provides no advantage to the adversary. The cost is an \(O(t)\) online transcript where \(t\) is the threshold.
This is the price of making every server individually accountable and
keeping the sharing state amenable to periodic renewal.  For long-lived
services that must survive gradual compromise over years, the protocol
purchases operational longevity that shorter-lived designs cannot offer.

\subsection*{Acknowledgments}
Vikas Srivastava acknowledges the support received from the ANRF-PMECRG project with Ref.
\texttt{ANRF/ECRG/2025/002808/PMS} and NIT Warangal Research Seed Grant. 

\bibliographystyle{unsrtnat}
\bibliography{references}

\end{document}